\documentclass[12pt,english]{article}
\usepackage{mathptmx}

\usepackage[T1]{fontenc}
\usepackage[margin=1.25in]{geometry}
\usepackage{tikz}
\usepackage{amsmath}
\usepackage{amsthm}
\usepackage{amssymb}
\usepackage{graphicx}
\usepackage{setspace}
\usepackage[authoryear]{natbib}
\usepackage{xcolor}
\usepackage{hyperref}
\hypersetup{
    colorlinks=true,
    linkcolor=blue,
    citecolor=magenta,      
    urlcolor=cyan,
    pdftitle={Overleaf Example},
    pdfpagemode=FullScreen,
    }
\usepackage{cleveref}
\makeatletter

\providecommand{\tabularnewline}{\\}

\theoremstyle{definition}
\newtheorem{defn}{\protect\definitionname}
\theoremstyle{plain}
\newtheorem{assumption}{\protect\assumptionname}
\theoremstyle{plain}
\newtheorem{lem}{\protect\lemmaname}
\theoremstyle{plain}
\newtheorem{prop}{\protect\propositionname}
\theoremstyle{plain}
\newtheorem{thm}{\protect\theoremname}
\theoremstyle{plain}

\usepackage{babel}

\providecommand{\assumptionname}{Assumption}
\providecommand{\definitionname}{Definition}
\providecommand{\lemmaname}{Lemma}
\providecommand{\propositionname}{Proposition}
\providecommand{\theoremname}{Theorem}
\providecommand{\corollaryname}{Corollary}
\makeatother

\usepackage{babel}
\providecommand{\assumptionname}{Assumption}
\providecommand{\definitionname}{Definition}
\providecommand{\lemmaname}{Lemma}
\providecommand{\propositionname}{Proposition}
\providecommand{\theoremname}{Theorem}

\usepackage{titlesec}
\titlespacing\section{0pt}{6pt}{4pt}
\titlespacing\subsection{0pt}{4pt}{2pt}
\titlespacing\subsubsection{0pt}{2pt}{2pt}
 \titlespacing*{\paragraph}{0pt}{1.25ex plus 1ex minus .2ex}{0.5em}

\begin{document}

\title{Engagement Maximization}
\author{Benjamin H\'{e}bert (Stanford and NBER) and Weijie Zhong (Stanford)\thanks{The authors would like to thank Peter DeMarzo, Sebastian Di Tella,
Darrell Duffie, Emir Kamenica, and Michael Woodford for helpful comments,
and Ian Ball and Arjada Bardhi for excellent discussions. All remaining errors are
our own. bhebert@stanford.edu and weijie.zhong@stanford.edu}}
\maketitle
\begin{abstract}
% We study a Bayesian agent receiving signals over time and then acting.
% The agent chooses when to stop and act, and prefers to act earlier
% all else equal. The signals are determined by a principal, who maximizes
% engagement (the total attention paid by the agent). We show that engagement
% maximization minimizes the agent's welfare,  induces excessive information
% acquisition (relative to an agent-optimal benchmark), and leads to
% extreme beliefs. The principal optimally sends only ``suspensive
% signals'' that lead the agent to become ``less certain than the
% prior'' and ``decisive signals'' that lead the agent to stop immediately.

We investigate the management of information provision to maximize user engagement. A principal sequentially reveals signals to an agent who has a limited amount of information processing capacity and can choose to exit at any time. We identify a ``dilution'' strategy---sending rare but highly informative signals---that maximizes user engagement. The platform's engagement metric shapes the direction and magnitude of biases in provided information relative to a user-optimal benchmark. Even without intertemporal commitment, the platform replicates full-commitment revenue by inducing the user's belief to remain ``as uncertain as'' the prior until the rare, decisive signal arrives and induces stopping. We apply our results to two contexts: an ad-supported internet media platform and a teacher attempting to engage test-motivated students.
\end{abstract}
Key Words: Information Acquisition, Recommendation Algorithms, Polarization,
Rational Inattention 
\noindent \begin{flushleft}
JEL Codes: D83, D86 
\par\end{flushleft}

\thispagestyle{empty}

\pagebreak{}

\setcounter{page}{1}

\section{Introduction\label{sec:Introduction}}

Free-to-use online platforms such as Facebook, Instagram, Youtube,
and Pinterest are used by billions of people worldwide and earn substantial
profits by displaying advertisements to their users. Their business
models are powered by personalized recommendation algorithms that
seek to maximize the ``engagement'' of each user by selectively
displaying content (\citet{lada_wang_yan_2021}, \citet{news_feed_engagement}).
Large quantities of computing and other resources have been invested
by these firms to develop algorithms that can predict billions of
users' preferences in real time.

A key challenge for these algorithms is to manage users' incentives.
Content is presented sequentially in ``news feeds,'' ``recommendations,''
or ``timelines,'' and users can choose freely \emph{what} to pay
attention to and \emph{when} to stop using the platform based on the
entire history of content presented thus far. From the platform's
perspective, providing the most useful content immediately is suboptimal,
as a user might become satisfied and choose to stop using the platform,
limiting the quantity of advertisements the platform can display to
the user. On the other hand, if the user anticipates that the platform
will never provide useful content, they will never begin to use the
platform in the first place.

A second, subtler challenge is to manage the algorithm's own ``incentives''. As is suggested by the previous paragraph, to maximize user's \emph{ex ante} engagement, platforms often design these algorithms with an initial commitment of providing reliable and wide-ranging content. Subsequent training of the algorithms based on user interactions, however, can induce algorithms to behave in a dynamically inconsistent manner. For instance, at a later time, the algorithm may ``learn'' to steer users into niche ``rabbit holes '' (\cite{news_tiktok}) because it delivers short-term engagement gains. When the platform lacks long term commitment power, this cause users to question the credibility of the platform's stated design goals. It is therefore unclear, ex-ante, whether concerns about dynamic inconsistency meaningfully constrain algorithm design.

Engagement maximization is relevant in many other contexts. We will pair the example of an internet platform engaging a user with a second example, that of a teacher who seeks to maximize
the engagement of a student who cares only about passing a test.\footnote{We are grateful to Emir Kamenica for suggesting this alternative setting
as an application of our model.} Our theoretical framework applies to this setting without modification. 

We study the optimal design of sequential information
presentation from a principal-agent perspective. The principal (the
platform/teacher) provides the agent (the user/student) with information. The agent
derives value from the information, and chooses when to stop engaging
with the provided information.\footnote{We show that it is without loss of generality to assume the agent
does not process the information selectively; that is, the agents
attends to the information the principal provides in equilibrium.} The rate at which the agent can process information is limited, and
the agent experiences an opportunity cost of time spent on the platform
net of any utility from using the platform. The principal's goal is
to maximize the attention the agent allocates to the platform (which
under some conditions will be equivalent in equilibrium to maximizing
the stopping time), and the principal accomplishes this by choosing
the nature of the information provided to the agent. The key modeling
assumptions we impose are that the principal knows perfectly the agent's
preferences and that the principal can flexibly manipulate the entire
information flow.

Two modeling features allow us to derive the rich implications in this environment. First, our model allows the measures of (i) engagement, (ii) information processing capacity, and (iii) the value of information to differ. This allows our model to capture a rich set of possible incentive misalignments between the principle and the agent. Second, we analyze our model both with and without intertemporal commitment, which allows us to characterize the implications of a requirement for dynamic consistency on the dynamics of information provision. 

Our main result is the characterization of an optimal strategy for
the principal. First, we completely characterize the optimal ``overall''
information structure (i.e. the signal-state joint distribution that
describes the beliefs the agent will hold when choosing to stop). It is characterized by the solution to
an augmented static rational inattention (RI) problem: the information
structure maximizes a linear combination of the instrumental value, the engagement measure, and the informativeness measure of information with endogenous weights.
Second, we identify one optimal sequential information structure for
the principal: the principal sends a ``dilution'' of the overall
information structure (a compound Poisson process such that a signal
arrives at a Poisson rate and upon arrival, the signal is distributed
according to the overall information structure). 

Having described a single optimal strategy, we then show that the following features must hold for \emph{any} optimal strategy:

\begin{itemize}
    \item \textbf{Engagement-driven biases}: We compare the optimal information structure with the benchmark solution when the agent chooses both the information and the stopping time. We show that engaging with the platform leads to the agent being biased towards ``more extreme'' beliefs relative to the agent-optimal benchmark. More specifically, we show that the direction and magnitude of such biases are determined by the principal's engagement measure.

    \emph{Leading Case: All Engagement is Profitable}. Suppose that the principal's engagement measure is identical to the informativeness measure. In the internet platform context, we interpret this case as one in which the agent is exposed to ads in proportion to the context she receives, with no measurement of whether she attends to the ads. In this case, the optimal information structure is narrow but deep: the principal only provides information about the agent's interested dimensions. This information is provided in a rare but chunky way, leading the agent to spend additional time and acquire additional information relative to the agent-optimal benchmark. This result echoes the observation that engagement maximizing recommendation algorithms often lead users down ``rabbit-holes.'' 

    \emph{Leading Case: Decision-Irrelevant Engagement is Profitable}. Now suppose that the principal's revenue comes exclusively from engagement with information that is irrelevant for decision. In the internet platform context, we interpret this case as one in which the agent is exposed to decision-irrelevant ads, and the principal benefits only when the agent engages with these ads (e.g. if ad revenue comes from the click of an ad). In this case, we show that the optimal information structure is broad but noisy: the principal provides excess information (e.g. sponsored content) that the agent views as pure noise, intermixed with just enough decision-relevant information to keep the agent engaged. This result echoes the observation that engagement-maximizing recommendation algorithms often intersperse useful content with ``click-bait'' and ``stealth marketing'' material. 
    
    \item \textbf{Credibility-driven dynamics}: To study policies that are credible under limited intertemporal commitment, we introduce a notion of subgame perfection to our model. Firstly, we show that commitment has no value at all: the principal can achieve the same payoff from the optimal commitment policy even under limited commitment. However, limited commitment has strong implication on the dynamics of the policy: an optimal policy corresponds to a subgame perfect equilibrium if (and nearly only if) the induced belief process of the agent jumps between a small set of very special beliefs --- beliefs at which the agent is ``as uncertain as'' the prior belief. As an implication, the optimal and credible signal process must be a generalized ``dilution'': belief stays constant until a signal arrives at a Poisson rate. The signal brings the posterior belief either to the stopping region or to another such special interim belief. Consequently, the ``dilution'' policy we constructed earlier is indeed credible, and is the uniquely credible policy in certain symmetric settings.

    The intuition for the result is simple: We show that the principal-optimal policy leaves the agent with no surplus at the prior. If the agent were to become ``more certain than'' the prior at an interim stage, the surplus from future information would be negative, and the agent would stop engaging with the principal. On the other hand, if there agent were to become ``less certain than'' the prior at an interim stage, her surplus from future information would be positive. At this point, the principal would be tempted to adjust the information process and extract more surplus, at the agents' expense. Thus, the only possible interim beliefs are those that are exactly ``as certain as'' the prior belief. Our main result in Section \ref{sec:Dynamics} is a formalization of this idea, along with a precise definition of what it means to be more certain, less certain, or as certain as the prior.
\end{itemize}

Our two leading cases are relevant to both the internet platform and teacher-student contexts. In the latter, they map to the degree to which the teacher values the student learning about test-relevant vs. test-irrelevant information. A teacher who values student learning but not test performance will provide just enough test-specific information to engage a student who cares only about the test, while providing additional content that is not part of the test. In contrast, a teacher who values all learning, regardless of whether it is covered on a test, will optimally provide only test-relevant information while inducing students to learn more test-relevant information than they would choose to learn on their own.

\subsection{Related Literature}

Our paper contributes to several strands of literature on the dynamic
provision of information. Viewing our principal as a media company,
our model is related to work on models of media bias (see \citet{gentzkow2015media}
for a survey). We share with \citet{kleinberg2022challenge} an interest
in explaining why the users of internet platforms would engage heavily
with those platforms while perceiving themselves as gaining little
from doing so. We derive this outcome as a result of strategic behavior
by rational agents with conflicting incentives; those authors emphasize
the time-inconsistency of user preferences. We share with \citet{acemoglu2021misinformation}
an emphasis on explaining what kind of information is available on
internet platforms; our analysis focuses on content selection algorithms,
whereas their analysis focuses on information sharing between users.

Closely related to our work is the literature on dynamic Bayesian
persuasion (e.g. \citet{ely2017beeps,renault2017optimal,ely2020moving,orlov2020persuading,che2020keeping})
that build on the static model of \citet{kamenica2011bayesian}. Our
setting differs from most of these papers in that our principal's
payoff solely depends on the engagement of the agent and not on her ultimate choice of action.\footnote{The leading case of our model in which all engagement is profitable is equivalent to one where the principal's only goal is
to ensure the agent continues to pay attention, as in \citet{kawamura2019news}.
However, our principal presents unbiased information, and hence is
not engaged in cheap talk (\citet{crawford1982strategic,cheng2022bayesian}).} Most of these papers focus on the ``Bayesian persuasion'' settings
where the principal's payoffs directly depend on the agent's action
and the maximization of ``engagement'' is a side effect. A more recent strand of literature focuses on the direct maximization of ``attention'' or ``engagement'' without the capacity constraint, including two concurrent papers (\citet{jan2020dynamic} and \citet{koh2022attention}) and two subsequent papers (\citet{koh2024persuasion} and \citet{saeedi2024getting}).\footnote{\citet{jan2020dynamic}'s main focus is on the competition between multiple senders. The single sender case of \citet{jan2020dynamic}
is the closest to ours. \citet{koh2022attention} and \citet{koh2024persuasion} focuses on more general implementability characterization. \citet{saeedi2024getting} focuses on the case with asymmetric prior.} These papers share several common predictions with our model, including Poisson dilution signals and the minimization of ex ante user welfare; indicating the robustness of these features across models. However, the absence of an informational constraint in these papers also drives stark differences---these papers predict full revelation of the state upon signal arrival, indicating no bias from the agent's preferred posterior beliefs.  Our analysis emphasizes the beliefs the agent will arrive at, and how these differ
between the principal's optimal strategy and an agent-optimal benchmark,
a comparison that is not possible absent information processing constraints.
We elaborate further on the connection between our model and these
models in Section \ref{subsec:Optimal-Policy-without}. We assume
both players are long-lived and have commitment power, in contrast
to the limited commitment settings in \citet{orlov2020persuading,che2020keeping}
or the myopic agent settings in \citet{ely2017beeps,renault2017optimal}.
Interestingly, unlike the findings in these papers, neither commitment
power nor forward-lookingness of the agent is necessary for sustaining
our equilibrium strategy (see Section \ref{ssec:credibility}).

Our model predicts gradual information revelation over time, which
is a feature shared by many of the dynamic information design models
(e.g. \citet{ely2015suspense,horner2016selling,che2020keeping,orlov2020persuading}).
However, unlike these papers, the gradual nature of belief evolution
in our model arises from the agent's information processing constraint,
as opposed to a desire to maximize suspense or address problems of
limited commitment. In particular, \citet{che2020keeping} predicts
that the agent's belief involves Poisson jumps and a drift, while
our optimal strategy admits Poisson jumps without a drift.

Formally, our approach is a principal-agent version of \citet{hebert2019rational}.
Those authors consider a model in which a single decision maker chooses
both what information to acquire and when to stop and act, whereas
in our model the principal chooses the information and the agent chooses
when to stop and act. We compare our model to a benchmark in which
the agent chooses both the information and when to stop and act; this
benchmark is characterized by results found in \citet{hebert2019rational}.
We follow \citet{hebert2019rational} in assuming that the principal
can choose any stochastic process for the agent's beliefs, subject
only to the martingale requirement (which is imposed by Bayesian updating)
and the upper bound on the agent's attention. We model this upper
bound using a ``uniformly posterior-separable'' information cost,
in the terminology of the rational inattention literature (\citet{caplin2017shannon}).\footnote{Examples of such information costs include mutual information, as
applied in \citet{sims2010rational}, as well as other proposed alternatives
(\citet{hebert2018informationcosts}, \citet{bloedel2020cost}).}

The rest of the paper is organized as follows. We begin in section
\ref{sec:The-Environment} by describing the basic environment of
our model. Section \ref{sec:Optimal-Policy} characterizes optimal
policy in our baseline model, with an emphasis on the beliefs the agent will hold when stopping. Section \ref{sec:Dynamics} discusses the dynamics of beliefs and the role of commitment
Section \ref{subsec:Optimal-Policy-without} discusses an extension of our baseline model
, and in Section \ref{sec:Conclusion} we conclude.

\section{\label{sec:The-Environment}The Environment}

\subsection{The Agent's Problem}

We study the problem of a rational, Bayesian agent receiving signals
about an underlying state for the purpose of taking an action. We
model information acquisition as a continuous time process, building
on results in \citet{hebert2019rational} and \citet{zhong2017jumps}.

Let $X$ be a finite set of possible states of nature. The state of
nature is determined ex-ante, does not change over time, and is not
known to the agent. Let $q_{t}\in\mathcal{P}(X)$ denote the agent's
beliefs at time $t\in[0,\infty)$, where $\mathcal{P}(X)$ is the
probability simplex defined on $X$. We will represent $q_{t}$ as
vector in $\mathbb{R}_{+}^{|X|}$ whose elements sum to one, each
of which corresponds to the likelihood of a particular element of
$X$, and use the notation $q_{t,x}$ to denote the likelihood under
the agent's beliefs at time $t$ over the true state being $x\in X$. For each state $x$, we use $e_x\in\mathcal{P}(X)$ to denote the degenerate belief that state $x$ occurs with probability $1$.

At each time $t$, the agent can either stop and choose an action
from a finite set $A$, or continue to acquire information. Let $\tau$
denote the time at which the agent stops and makes a decision, with
$\tau=0$ corresponding to making a decision without acquiring any
information. The agent receives utility $u_{a,x}$ if she takes action
$a$ and the true state of the world is $x$, and pays a flow cost
of delay per unit time, $ \kappa >0$, until an action is taken.\footnote{Note that we assume that the delay cost enters the agent's utility additively and the agent does not discount the utility from
acting. Discounting, in our context, has the effect of changing the
relative values of acting and never acting, and for this reason would
complicate our analysis; see the more extensive discussions in \citet{hebert2019rational}
and \citet{zhong2017jumps} on this point.}

Let $\widehat{u}(q')$ be the payoff (not including the cost of delay)
of taking an optimal action under beliefs $q'\in\mathcal{P}(X)$:
\[
\widehat{u}(q')=\max_{a\in A}\sum_{x\in X}q'_{x}u_{a,x}.
\]
In what follows, the convex function $\widehat{u}:\mathcal{P}(X)\rightarrow\mathbb{R}$
will summarize the value the agent places on information. This function,
as opposed to the utility function $u_{a,x}$, can be viewed as the
primitive ``value of information when acting'' in our model; nothing
in our analysis will depend on the nature of the action space $A$
or on the predicted joint distribution of states and actions. As a
result, our analysis extends without modification to the case in which
the agent values information for its own sake (i.e. for non-instrumental
reasons).

The agent's beliefs, $q_{t}$, will evolve as a martingale. This property
follows from ``Bayes-consistency.'' In a single-period model, Bayes-consistency
requires that the expectation of the posterior beliefs be equal to
the prior beliefs. The continuous-time analog of this requirement
is that the belief process is a martingale.

Formally, let $\Omega$ be the sample space. Let $q:\Omega\times\mathbb{R}_{+}\rightarrow\mathcal{P}(X)$
be a canonical c\`{a}dl\`{a}g stochastic process on $\mathcal{P}(X)$, let $\{\mathcal{F}_{t}\}$
be the natural filtration associated with this canonical process,
and let $\mathcal{F}=\lim_{t\rightarrow\infty}\mathcal{F}_{t}$. Let
$\mathcal{T}$ be the set of
non-negative stopping times with respect to $\{\mathcal{F}_{t}\}$.
%\subset\Omega^{\mathbb{R}_{+}}

Given a probability measure $P$ defined on $(\Omega,\mathcal{F})$,
$(\Omega,\mathcal{F},\{\mathcal{F}_{t}\},P)$ defines a probability
space. The agent's problem, given this probability space, is to choose
a stopping time to solve 
\[
V(P)=\sup_{\tau\in\mathcal{T}}\mathbb{E}^{P}[\widehat{u}(q_{\tau})- \kappa \tau|\mathcal{F}_{0}].
\]

\subsection{The Principal's Problem}

The principal chooses the information the agent receives so as to
maximize engagement. We begin with defining a measure for the agent's information processing
using a continuous time version of what \citet{caplin2017shannon}
call a ``uniformly posterior-separable'' cost function, as described
in \citet{hebert2019rational}. Uniformly posterior-separable cost
functions are defined in terms of a ``generalized entropy function,''
$H:\mathcal{P}(X)\rightarrow\mathbb{R}_{+}$. We assume that $H$
is twice continuously-differentiable and strongly convex.\footnote{I.e. the Hessian matrix of $H$ is strictly positive definite.} 
We assume the agent can process information
at a rate no greater than $\chi > 0 $. That is, the process $H(q_t)-\chi t$ must be an $\mathcal{F}_t$-supermartingale, requiring that for all $t$ and $\delta>0$, 
\begin{equation}
\mathbb{E}^P[H(q_{t+\delta})-\chi(t+\delta)|\mathcal{F}_t] \leq H(q_t) - \chi t ,\label{eq:general-constraint}
\end{equation}
where we adopt the convention that $q_t = q_\tau$ for $t>\tau$.\footnote{This constraint is equivalent to the one studied in \citet{hebert2019rational} within the class of uniformly posterior-separable costs, allowing us to invoke their results.}

The principal's goal is to maximize ``engagement,'' which is not necessarily the same as maximizing the information acquired.
Specifically, we assume that the principal earns profits in proportion
to cumulative information acquisition as measured by the convex function $G:\mathcal{P}(X)\rightarrow\mathbb{R}_{+}$. 
The function
$G$ is not necessarily strictly or strongly convex, although $G=H$ will be one of our leading cases.

%We will assume that there is always some information the principal would like the agent to acquire (i.e. that the Hessian of $G$ is non-zero on the interior of the simplex).

The principal's goal is to design the agent's belief process so as
to maximize profits, taking into account the fact the agent will know
the nature of the beliefs process and optimally choose when to stop
paying attention. Let $\bar{q}_{0}\in\mathcal{P}(X)$ be the agent's
prior. The principal chooses his policies from the set $\mathcal{A}(\bar{q}_{0})$,
which is the set of probability measures on $(\Omega,\mathcal{F})$
such that $q$ is martingale belief processes with $q_{0}=\bar{q}_{0}$ and
non-negative stopping times $\tau$ such that \eqref{eq:general-constraint} is satisfied.
Formally, the principal chooses the probability measure $P$, which
is equivalent to choosing the law of the belief process $q$. In
this sense, the principal can choose any c\`{a}dl\`{a}g martingale
belief process with $q_{0}=\bar{q}_{0}$, subject to the constraint
imposed by the agent's information processing capacity.\footnote{An alternative interpretation is that the principal can send more
information than the agent can process, in which case the agent chooses
which information to attend to. Allowing the agent this choice cannot
benefit the principal, and it is therefore without loss of generality
to suppose the principal chooses a process that satisfies the agent's
information processing constraint.}
\begin{defn}
\label{def:The-principal's-problem}The principal's problem given
initial belief $\bar{q}_{0}\in\mathcal{P}(X)$ is to maximize engagement,
\begin{align}
J(\bar{q}_{0}) & =\sup_{(P,\tau)\in\mathcal{A}(\bar{q}_{0})}\mathbb{E}^{P}[G(q_{\tau})-G(\bar{q}_{0})|\mathcal{F}_{0}]\tag{P}\label{eq:P}
\end{align}
subject to the agent's stopping decision, 
\[
\mathbb{E}^{P}[\widehat{u}(q_{\tau})- \kappa \tau|\mathcal{F}_{0}] \geq \sup_{\tau'\in\mathcal{T}}\mathbb{E}^{P}[\widehat{u}(q_{\tau'})- \kappa \tau'|\mathcal{F}_{0}].
\]
\end{defn}

\subsection{Discussion}

The generalized entropy functions $G$ and $H$ and value of information $\hat{u}$ play conceptually
distinct roles in our model. The function $H$ governs the agent's
information processing constraint. It determines which kinds of information
are relatively easy or difficult for the agent to process. The function $\hat{u}$ describes the information the agent would like to receive, and the function $G$ describes the information the principal would like the agent to receive.

To simplify the discussion that follows, we will assume that the state space has a product structure, $X = X_1 \times X_2$, and that the agent values only information about $x_1 \in X_1$ (i..e $\hat{u}$ depends only on the marginal posterior over $X_1$). The agent benefits from learning about $x_2\in X_2$ only if doing so makes it easier to learn about $x_1 \in X_1$. Whether or not this is the case is governed by the $H$ function, and our general framework admits both possibilities.\footnote{The specific property of the $H$ function that governs whether or not learning about $x_2$ makes it cheaper to learn about $x_1$ is what \cite{hebert2023information} call ``R-monotonicity.''} We interpret $x_1$ as the useful content and $x_2$ as content that has no intrinsic value to the user. Throughout the paper, we assume that the user ``passively'' learns about both $x_1$ and $x_2$, as long as the principal supplies both types of content. In the internet platform context, this can be interpreted as reflecting ``stealth marketing'' or ``click-baiting'' behavior that makes ads indistinguishable from relevant content. However, even if the ads are obvious and the user can effortlessly avoid them, she still finds it optimal to fully process information about $x_2$ in equilibrium.\footnote{This statement has been formally proved in the working paper version of the current paper (see Proposition 8 of \cite{hébert2025engagementmaximization}). }

We focus our discussion on the extent to which $G$ differs from $H$ and $\hat{u}$. One leading case of interest is when $G=H$.\footnote{The case in which $G$ is proportional to $\hat{u}$ is essentially identical to the $G=H$ case, and for this reason we will discuss only the latter.} We interpret this case as reflecting a kind of indifference on the part of the principal: he does not care what the agent learns about, only that she learn as much as possible. In equilibrium, the constraint \eqref{eq:general-constraint} will bind; as a result, ``learning as much as possible'' is equivalent to ``spending as much time learning as possible.''\footnote{The authors would like to thank Emir Kamenica for suggesting this interpretation. Note that although time spent and information processed are equivalent when \eqref{eq:general-constraint} binds, they are not equivalent more generally (e.g. if the principal sends no signals to waste time), which is why we describe the equivalence as ``time spent learning.''} We will show in this case that the signals sent by the principal maximize the utility of a hypothetical agent with a lower cost of delay.

A second leading case is one in which $G$ measures only information acquisition that is irrelevant to the agent's decision problem (i.e. it depends only on the marginal posterior of $X_2$). We interpret this case as capturing a multi-dimensional conflict between the principal and agent over both how much the agent should learn and what the agent should learn about. In this case, the signals sent by the principal are not equivalent to those that would be chosen by a hypothetical agent with a lower cost of delay; the principal also distorts the signals so that they are more informative about $X_2$. In the special case in which $H$ is proportional to Shannon's entropy, we derive a sharp result: the principal provides the agent with information about $x_1 \in X_1$ that is identical to what the agent would optimally choose for herself, plus additional information about $x_2\in X_2$ that the agent would never choose to receive.

There are also intermediate situations of interest. The information the principal would like the agent to acquire might be of some use to the agent without being neutral or aligned with what the agent would like to learn (i.e. $G$ can depend in part of the posterior over $X_1$ without being a linear combination of $\hat{u}$ and $H$). Our two leading cases generate sharp predictions precisely because they are extreme.

\section{\label{sec:Optimal-Policy}Optimal Policy}

We start by defining a relaxed version of the principal's problem.
Any probability measure and stopping rule the principal can implement
will induce a probability measure over beliefs the agent will hold
when she chooses to stop (i.e. a law for $q_{\tau}$). Define $\Pi(\bar{q_0}) = \lbrace \pi \in \mathcal{P}(\mathcal{P}(X)): \mathbb{E}^{\pi}[q]=\bar{q}_0 \rbrace$ as the set of probability measures over posterior beliefs that are consistent with the initial prior. Under any feasible policy $(P,\tau)$ chosen by the principal, the law of the stopped belief $q_{\tau}$ induced by this policy will be an element of $\Pi(\bar{q}_0$), by the martingale property of beliefs.

The following lemma describes the date-zero participation constraint
of the agent (i.e., incentive compatibility w.r.t. stopping at time 0) and an upper bound on the total engagement. 
\begin{lem}
\label{lem:neccesary} $\forall(P,\tau)\in\mathcal{A}(\bar{q}_0)$
satisfying the agent's optimal stopping constraint in Definition \ref{def:The-principal's-problem}, the law $\pi$ of the stopped belief $q_{\tau}$ satisfies: 
\begin{enumerate}
\item $\mathbb{E}^{\pi}[\widehat{u}(q)] \geq \kappa \mathbb{E}^{P}[\tau|\mathcal{F}_{0}] + \widehat{u}(\bar{q}_{0})$,
and 
\item $\mathbb{E}^{\pi}[H(q)-H(\bar{q}_{0})]\le\chi\mathbb{E}^{P}[\tau|\mathcal{F}_{0}]$. 
\end{enumerate}
\end{lem}
\begin{proof}
    See Appendix \ref{ssec:proof:necessary}.
\end{proof}

Lemma \ref{lem:neccesary} presents a necessary condition for any
admissible policy for the principal in the optimization problem (Definition
\ref{def:The-principal's-problem}). The first condition states that
the agent's optimal stopping utility is weakly greater than the utility
from stopping immediately. The second condition states that the expected cumulative
information acquired by the agent is weakly less than $\chi\mathbb{E}[\tau|\mathcal{F}_{0}]$---the
maximal attention permitted by the information constraint \eqref{eq:general-constraint}.

Combining these two constraints, 
\[
\mathbb{E}^{\pi}[\widehat{u}(q)-\widehat{u}(\bar{q}_{0})]\geq \kappa \mathbb{E}^{P}[\tau|\mathcal{F}_{0}]\geq\frac{ \kappa }{\chi}\mathbb{E}^{\pi}[H(q)-H(\bar{q}_{0})].
\]
Let us define the principal's relaxed optimization problem as a maximization over $\Pi(\bar{q}_0)$ incorporating
only this combined constraint: 
\begin{align}
\bar{J}(\bar{q}_0) & =\sup_{\pi\in \Pi(\bar{q}_0)} \mathbb{E}^\pi[G(q) - G(\bar{q}_0)]\tag{R}\label{eq:p:relaxed}\\
s.t.\  & \frac{ \kappa }{\chi}\mathbb{E}^{\pi}[H(q)-H(\bar{q}_{0})]\le\mathbb{E}^{\pi}[\widehat{u}(q)-\widehat{u}(\bar{q}_{0})].\nonumber 
\end{align}
Because this combined constraint must hold in the original principal's
problem, we must have $\bar{J}(\bar{q}_{0})\ge J(\bar{q}_{0})$.

We assume that (1) it is possible for the principal to benefit from engagement and (2) it is prohibitively costly for the agent to learn enough information to achieve the principal's first best, to the point that
the agent would prefer to learn nothing at all if confronted with
only these two possibilities.

\begin{assumption}
\label{ass:imperfect} The principal can benefit from engagement: $\max_{\pi' \in \Pi(\bar{q}_0)} \mathbb{E}^{\pi'}[G(q)]> G(\bar{q}_0)$. However, it is not possible for the principal to achieve the principal's first best: if $\pi \in \arg \max_{\pi' \in \Pi(\bar{q}_0)} \mathbb{E}^{\pi'}[G(q)]$, then 
\begin{align*}
    \widehat{u}(\bar{q}_{0})-\frac{ \kappa }{\chi}H(\bar{q}_{0})> \mathbb{E}^{\pi}\left[\widehat{u}(q)-\frac{\kappa}{\chi}H(q)\right].
\end{align*}
\end{assumption}
This rules out the possibility that the principal either has no desire for the agent to acquire information or the possibility of satiating this desire. It immediately implies that the constraint in \eqref{eq:p:relaxed} will bind.

Define the function $\mathcal{L}$ as the Lagrangian associated with the dual of \eqref{eq:p:relaxed}, 
\[
\mathcal{L}(\pi,\lambda)= \mathbb{E}^{\pi}\left[\widehat{u}(q)-\frac{ \kappa }{\chi}H(q)+\lambda G(q)\right].
\]

The following proposition characterizes the solutions of \eqref{eq:p:relaxed}.
\begin{prop}\label{prop:existence}
A solution to \eqref{eq:p:relaxed} exists. Additionally,
\begin{itemize}
\item there exists a $\lambda\geq0$ such that if $\pi^*$ is a solution to \eqref{eq:p:relaxed} then  
\[\pi^* \in \arg \max_{\pi\in \Pi(\bar{q}_0)} \mathcal{L}(\pi,\lambda), 
\]
\item if $\pi^*$ is a solution to \eqref{eq:p:relaxed}, then the constraint binds, 
\[\frac{ \kappa }{\chi}\mathbb{E}^{\pi^{*}}[H(q)-H(\bar{q}_{0})]=\mathbb{E}^{\pi^{*}}[\widehat{u}(q)-\widehat{u}(\bar{q}_{0})],
\] 
\item if $\pi^*$ is a solution to \eqref{eq:p:relaxed}, then starting \eqref{eq:p:relaxed} from any belief in the support of $\pi^*$, it is not possible to induce engagement, $\mathbb{E}^{\pi^*}[\bar{J}(q)]=0$, and 
\item at least one solution to \eqref{eq:p:relaxed} has finite support.
\end{itemize}
\end{prop}

\begin{proof}
  See Appendix \ref{ssec:proof:existence}.
\end{proof}

This proposition demonstrates that $\pi^{*}$ is the solution to a
static rational inattention problem, with a general UPS cost function
(of the sort studied by \citet{caplin2017shannon}). Those authors
show that a necessary condition for $\pi^{*}$ is that it \emph{concavifies}
the function $\widehat{u}-\frac{ \kappa }{\chi}H+\lambda G$.\footnote{Related results appear in earlier working papers by those authors
and in the Bayesian persuasion literature.} Here, $\hat{u}$ captures the benefit of information to the agent, $\frac{\kappa}{\chi} H$ captures the costly delay required to acquire the information, and $\lambda G$ captures the benefit of the information acquisition to the principal, converted to the units of the agent's utility at the rate $\lambda$.

Engagement is infeasible ($\bar{J}(\bar{q}_0)=0$) whenever the only
$\pi$ satisfying the constraint are ones with $\mathbb{E}^{\pi}[G(q)-G(\bar{q}_{0})]=0$,
which is to say that it is not possible for the principal to induce
information acquisition in the dimensions he cares about. If $G$
is strictly convex (meaning that all information acquisition is at
least somewhat beneficial to the principal, as in the $G=H$ case),
this implies $\text{Supp}(\pi)=\{\bar{q}_{0}\}$. When $\text{Supp}(\pi)=\{\bar{q}_{0}\}$,
we will say that $\pi^{*}$ is degenerate, and otherwise say that
$\pi^{*}$ is non-degenerate.\footnote{When $G$ is not strictly convex, meaning that there are some dimensions
of potential information acquisition the principal does not value,
it is possible to have $J(\bar{q}_{0})=0$ with a non-degenerate $\pi^{*}$.
However, the information acquired in this case generates no surplus
for either the principal or the agent.}

Let us take as given a solution $\pi^{*}$ to this relaxed problem,
and consider how it might be implemented in an incentive compatible
way in the original principal's problem.

\subsection{Implementation}

Take any non-degenerate $\pi\in \Pi(\bar{q}_0)$, and define the stochastic process
$q_{t}$ as: 
\begin{align}
q_{t}=\bar{q}_{0}+\mathbf{1}_{N_{\alpha}(t)\ge1}\cdot(Q-\bar{q}_{0}),\label{eq:Poisson}
\end{align}
where $Q\in\mathcal{P}(X)$ is a random variable distributed according
to $\pi$ and $N_{\alpha}(t)$ is an independent Poisson counting
process with parameter $\alpha$. Here, $q_{t}$ is a compound Poisson
process that jumps according to $\pi$ at rate $\alpha$. We call
such process $q_{t}$ an $\alpha-$\emph{dilution} of $\pi$.\footnote{\citet{pomatto2018cost} first introduce the notion of dilution. They
define the dilution of an information structure $\pi$ as ``producing
$\pi$ with probability $\alpha$ and uninformative signal with probability
$1-\alpha$.'' (the same notion appeared in \citet{bloedel2020cost}).
Our notion of $\alpha$-dilution is essentially the repetition of
a dilution in continuous time.} 
\begin{prop}
\label{prop:implementation} For all non-degenerate $\pi\in \Pi(\bar{q}_0)$
that satisfy the constraint in \eqref{eq:p:relaxed}, let $\alpha=\frac{\chi}{\mathbb{E}^{\pi}[H(q)-H(\bar{q}_{0})]}$
and let $q_{t}$ be the $\alpha$-dilution of $\pi$. Let $P$ be the law of $q_t$ and let $\tau = \inf\{t>0: q_t \neq \bar{q}_0\}$. Then $(P,\tau)$
is feasible in \eqref{eq:P}
and implements utility level $\mathbb{E}^{\pi}[G(q)-G(\bar{q}_{0})]$. 
\end{prop}
\begin{proof}
See Appendix \ref{ssec:proof:implementation}.
\end{proof}
Note that Proposition \ref{prop:implementation} immediately implies that $J=\bar{J}$; hence, we will not distinguish between the two functions. Combining Lemma \ref{lem:neccesary}, Proposition \ref{prop:existence},
and Proposition \ref{prop:implementation}, we obtain the main characterization
of the optimal policy: 
\begin{thm}
\label{thm:main-result}$\forall\bar{q}_{0}\in\mathcal{P}(X)$, there
exists a $\pi^{*}\in \Pi(\bar{q}_0)$ with finite support
solving \eqref{eq:p:relaxed}. If Supp$(\pi^{*})=\{\bar{q}_{0}\}$,
%the agent will immediately stop and 
any feasible policy is optimal.
Otherwise, let $\alpha^{*}=\frac{\chi}{\mathbb{E}^{\pi^{*}}[H(q)-H(\bar{q}_{0})]}$,
and let $(P^{*},\tau^{*})$ be the law and jumping time of the $\alpha^{*}$-dilution
of $\pi^{*}$. Then, $(P^{*},\tau^{*})$ solves the principal's problem \eqref{eq:P}. 
\end{thm}
There are generally many optimal policies in the principal's problem.
First, there may be multiple $\pi^{*}$ that solve the relaxed principal's
problem (although uniqueness in static rational inattention problems
is guaranteed under certain additional assumptions). Second, there
are many stochastic processes $q_{t}$ (equivalently, laws $P$) and
stopping times $\tau$ that induce the same law for $q_{\tau}$; provided
that this law is equal to $\pi^{*}$ and that incentive compatibility
and the bounds on information acquisition are satisfied, all such
policies are optimal. However, it is not the case that anything goes,
as we show in \cref{sec:Dynamics}.

An immediate implication of Theorem \ref{thm:main-result} is that
the agent's participation constraint binds. That is, 
\[
\mathbb{E}^{P^*}[\widehat{u}(q_{\tau^{*}})- \kappa \tau^{*}|\mathcal{F}_{0}]=\widehat{u}(\bar{q}_{0}).
\]
Strikingly, the agent is no better off receiving information from
an engagement-maximizing principal than if she could not receive any
information at all. The principal extracts the full surplus generated
by the ability to produce information and apply it in the agent's
decision problem, despite the agent's ability to choose when to stop
and act. The principal, in maximizing the engagement of the agent,
minimizes the agent's welfare subject to a participation constraint. Note, however, that the principal's ability to extract all of the surplus depends critically on the possibility of jumps in the belief process. We provide an example in Proposition \ref{prop:bound:binary} in which beliefs are required to be continuous and the agent-optimal outcome is attained.

\subsection{The agent-optimal benchmark and extreme beliefs}

We will compare our results to a benchmark in which the agent and
not the principal chooses the probability space and martingale belief
process (subject to \eqref{eq:general-constraint}). This benchmark
is a special case of the more general models described in \citet{hebert2019rational}
and \citet{zhong2017jumps}.

Those authors show that the optimal policies of this benchmark dynamic
model are also equivalent to the solution to a static rational inattention
problem. That is, under the agent-optimal policies $(P_A,\tau_A)$, given the initial belief
$\bar{q}_{0}\in\mathcal{P}(X)$, 
\begin{align}
\mathbb{E}^{P_A}[\widehat{u}(q_{\tau_A})- \kappa \tau_A|\mathcal{F}_{0}] & =V^{B}(\bar{q}_{0})\nonumber\\
 & =\max_{\pi\in \Pi(\bar{q}_0)}\mathbb{E}^{\pi}[\widehat{u}(q)-\frac{ \kappa }{\chi}(H(q)-H(\bar{q}_{0}))].\tag{B}\label{eq:agent-bench}\nonumber
\end{align}
This solution can be implemented, as above, by a compound Poisson
process, but also by a diffusion (and many other processes). The following characterization is straightforward.

\begin{lem}\label{prop:agent}
If the agent would choose to acquire information in the agent-optimal benchmark, the principal will induce engagement: for all $q \in \mathcal{P}(X)$, if $V^B(q)>\widehat{u}(q)$ then $\bar{J}(q)>0$. 
    
Moreover, if $\pi_A$ is a solution to \eqref{eq:agent-bench}, then starting \eqref{eq:agent-bench} from any belief in the support of $\pi_A$, it is weakly optimal for the agent to stop, $\mathbb{E}^{\pi_A}[V^B(q)-\hat{u}(q)]=0$. 

\end{lem}
\begin{proof}
    See Appendix \ref{ssec:proof:agent}.
\end{proof}

This lemma highlights two points. First, if the agent would choose to acquire information in the agent-optimal benchmark, the principal can profit by providing the agent with information (although that information will not in general be the information the agent would have chosen to acquire). Second, at any belief the agent might hold when stopping in the agent-optimal benchmark, the agent must have no benefit from continuing to acquire information.

The static rational inattention problem that characterizes stopping beliefs in the agent-optimal benchmark is essentially identical
to the one that characterizes optimal policy in our principal-agent
framework, except that the UPS information cost is $\frac{ \kappa }{\chi}H(q)$
in the benchmark problem and $\frac{ \kappa }{\chi}H(q)-\lambda G(q)$
in the principal agent case. The similarity in structure between these
two problems will allow us to highlight the implications of engagement
maximization.

Let $Q^{i}(\bar{q}_{0})\subseteq\mathcal{P}(X)$ be the union of the
support of all optimal stopping beliefs in the benchmark ($i=a$)
and principal-agent ($i=p$) models. Let $\mathrm{Conv} \ Q^{i}(\bar{q}_{0})$
denote the convex hull of $Q^{i}(\bar{q}_{0})$. The following proposition
demonstrates that the beliefs the agent will hold after
engaging with the principal are more extreme than the beliefs the
agent would choose to acquire in the benchmark model. This result
follows from the observation that in both models, stopping beliefs
are characterized by the solution to a static rational inattention
problem, with a lower information cost in the principal-agent case
than in the benchmark case. There are three possible exceptions to this
result. One is the case in which the principal can provide no strict incentive via information ($V^B(\bar
q_0)=0$), in which case $Q^{p}(\bar{q}_{0})\subset Q^a(\bar{q}_0)$. The second is the
case in which the agent-optimal beliefs lie on the boundary of the simplex,
in which case there may not be ``room'' for beliefs to become more
extreme. The third is the possibility that the marginal benefit to the principal of additional information acquisition jumps downward exactly at the agent-optimal posterior beliefs, which can be ruled out by assuming $G$ is differentiable.
\begin{prop}
\label{prop:extreme-beliefs} Assume that $G$ is continuously differentiable, $ V^B (\bar{q}_0)>\widehat{u}(\bar{q}_0)$, and
that $Q^{a}(\bar{q}_{0})$ is a subset of the relative interior of
$\mathcal{P}(X)$. Then for all $q\in Q^{p}(\bar{q}_{0}),q\not\in \ \mathrm{Conv} \ Q^{a}(\bar{q}_{0})$.
\end{prop}
\begin{proof}
See the appendix, section \ref{subsec:proof:extreme-belief}. 
\end{proof}

That is, the stopping beliefs in the principal-agent problem will lie outside the convex hull of the stopping beliefs of the agent-optimal benchmark. When $G$ is strictly convex, it is also possible to show that the stopping beliefs in the agent-optimal benchmark are in the relative interior of the convex hull of the stopping beliefs in the principal-agent problem. However, when the principal profits only by providing the agent with irrelevant information, $G$ is not strictly convex, and the agent-optimal benchmark stopping beliefs will not necessarily lie in the relative interior. Figure \ref{fig:6} below provides an example of this possibility.

Extreme beliefs are a natural consequence of engagement maximization. By forcing the agent to ultimately acquire more information
than she would choose for herself (the extreme beliefs), the principal
can simultaneously delay the agent's stopping decision while providing
information the agent is willing to attend to. Interpreted in the
context of an internet platform and a user, our
extreme beliefs result implies that the platform should provide in-depth
content on a narrower set of topics than would be preferred by the
user. We can interpret this loosely as encouraging the user to go
down ``rabbit holes'' that lead to extreme beliefs. In the teacher-student context, extreme beliefs have a more benign interpretation: the teacher induces test-motivated students to learn more than they would choose to learn on their own.

\subsubsection{Leading cases}\label{sssec:leading:case}

\cref{prop:extreme-beliefs} establishes that engagement maximization necessarily leads to more extreme stopping beliefs relative to the agent optimal benchmark. In what follows, we provide a more detailed analysis of the direction towards which engagement maximization biases the beliefs in two salient settings. To help interpretation, we assume that the state space has a product structure, $X=X_1 \times X_2$, and that the agent's decision problem $u$ depends only on $x_1 \in X_1$. We further suppose that $x_1$ and $x_2$ are independent under the prior $\bar{q}_0$, so that an agent who receives signals that depend only on $x_1$ will not update about $x_2$. Thirdly, we assume throughout this subsection that $H$ is the negative of Shannon's entropy $H^S$. 

\paragraph{All Engagement is Profitable.}

When the principal benefits from all forms of engagement by the agent (our interpretation of the $G=H$ case), it is never worthwhile for the principal to provide the agent with decision-irrelevant information. Instead, the principal induces the agent to acquire more decision-relevant information than the agent would choose to acquire on her own.

The following result shows that this is what happens in the principal-agent problem-- the principal will send the agent information about $x_1$ only. This occurs even though the principal would benefit from the agent learning about $x_2$ as well. Intuitively, signals about $x_1$ and $x_2$ that benefit the principal equally take an equal amount of time for the agent to process (because $G=H$), but only the former has a utility benefit for the agent. A principal who was sending signals about $x_2$ could switch to sending signals about $x_1$, relax the agent's participation constraint,\footnote{This argument depends on the assumption that sending a signal conditioned only on $x_1$ is least costly way to communicate information about $x_1$; that $H$ is proportional to Shannon's entropy ensures this.} and then profit by adjusting the signals so that they induced even more extreme beliefs about $x_1$ for the agent.

\begin{prop}\label{prop:attention}
    Suppose $G=H$. Then, \eqref{eq:p:relaxed} is solved by $\pi$ that reveals $x_1$ only, i.e. for all $q \in \text{Supp}(\pi)$  the marginals of $q$ and $\bar{q}_0$ on $X_2$ are identical.
    %and all functions $f:X_2 \rightarrow \mathcal{R}$, $\mathbb{E}^q [f(x_2)] =  \mathbb{E}^{\bar{q}_0} [f(x_2)]$.
\end{prop}
\begin{proof}
    See Appendix \ref{ssec:proof:attention}.
\end{proof}

\paragraph{Decision-Irrelevant Engagement is Profitable.}

Let us now consider a different leading case, in which the information that benefits the principal is irrelevant to the agent's decision. Maintaining the assumption of a product structure for the state space, an independent prior, and that $H$ is proportional to Shannon's entropy, we now suppose that $G$ depends only on $x_2$ (recall that $\widehat{u}$ depends only on $x_1$).

\begin{prop}\label{prop:noise}
   Suppose $G$ depends on $x_2$ only. Then, \eqref{eq:p:relaxed} is solved by $\pi=\pi^*_1\otimes \pi_2$, where $\pi^*_1\in \mathcal{P}(\mathcal{P}(X_1))$ is optimal for the agent and $\pi_2\in\mathcal{P}(\mathcal{P}(X_2))$.
\end{prop}
\begin{proof}
    See Appendix \ref{ssec:proof:noise}.
\end{proof}

When the principal cares only about providing decision-irrelevant information, the principal will provide a mix of decision-relevant and decision-irrelevant information. The decision-relevant information will be exactly the information that the agent would have chosen for herself, but the agent will be unable to avoid learning the decision-irrelevant information that is ``mixed in'' with the decision-relevant information.\footnote{We conjecture that the property required for the result is the $R$-mononticity property discussed in \cite{hebert2023information}.}

\subsection{Illustrative Examples\label{subsec:Examples}}

\paragraph*{Example 1.} We illustrate our results thus far in a
simple example. An agent faces a choice between two actions, $A=\{l,r\}$.
The payoffs from the actions are uncertain and depend on the state
of the world $x\in X=\{L,R\}$. The agent assigns equal prior probability
to both states ($\bar{q}_0$ is uniform). The agent gets utility one when the chosen action
matches the state ($u_{l,L}=u_{r,R}=1$) and utility negative one otherwise
($u_{l,R}=u_{r,L}=-1$). The agent is impatient and pays a constant cost
of delay of two utils per unit of time ($ \kappa =2$). We assume that the information processing constraint $H$ is the negative of Shannon's entropy $H^S$, and that the engagement measurement function $G$ is also the negative of Shannon's entropy. We assume $\chi=1$.

By Theorem \ref{thm:main-result}, the principal's optimal strategy maximizes 
\begin{align*}
    \mathbb{E}^{\pi}[\widehat{u}(q)-(2-\lambda)(H^S(0.5)-H^S(q))]
\end{align*}
for some $\lambda\in(0,2)$. Meanwhile, the agent's optimal strategy maximizes
\begin{align*}
    \mathbb{E}^{\pi}[\widehat{u}(q)-2(H^S(0.5)-H^S(q))].
\end{align*}
Both problems are static rational inattention problems. An application of \cite{matejka2015rational} implies that the agent's optimal stopping belief has support $\{\frac{e}{e+1},\frac{1}{e+1}\}$ and the principal's optimal stopping belief has support $\Big\{\frac{e^{2/(2-\lambda)}}{e^{2/(2-\lambda)+1}},\frac{1}{e^{2/(2-\lambda)}+1}\Big\}$. Both probability measures of stopping beliefs can be implemented by dilutions (compound Poisson processes). Since $\lambda\in(0,1)$, it is easy to see that the principal induces more extreme posterior beliefs and a longer wait compared to the agent-optimal benchmark.

Next, we vary the agent's prior belief $\bar{q}$, and focus our analysis on the agent's posteriors. Figure \ref{fig:1} illustrates the
optimal $q_{\tau}$ for every prior belief: When $\bar{q}_{0}\le\underline{q}$
or $\bar{q}_{0}\ge\bar{q}$, the agent stops immediately no matter
what information she will receive. When $\bar{q}_{0}\in(\underline{q},\bar{q})$,
the dashed lines illustrate the support of the agent-optimal policy.
The lines are flat, meaning that the agent's stopping beliefs do not
change with the prior belief, consistent with the prediction of the
standard RI models.\footnote{This property is known as the ``locally invariant posteriors'' (\citet{caplin2017shannon}).}
The solid curves illustrate the support of the unique optimal $\pi^{*}$.
They are closer to the boundaries zero and one, indicating that the
posterior beliefs become more extreme under the principal's optimal
policy. 
\begin{figure}[htbp]
\centering 
\tikzset{every picture/.style={line width=0.75pt}} %set default line width to 0.75pt        

\begin{tikzpicture}[x=0.75pt,y=0.75pt,yscale=-1,xscale=1]

%Straight Lines [id:da8562511125186589] 
\draw    (70,250) -- (380,250) ;
%Straight Lines [id:da6057405459295908] 
\draw    (70,250) -- (70,50) ;
%Straight Lines [id:da8775142730328068] 
\draw    (70,250) -- (370,70) ;
%Straight Lines [id:da6970512733632614] 
\draw    (70,50) -- (76,50) ;
%Straight Lines [id:da6954445864400339] 
\draw    (70,130) -- (76,130) ;
%Straight Lines [id:da03331176162122673] 
\draw    (70,210) -- (76,210) ;
%Straight Lines [id:da3518100134592135] 
\draw    (70,170) -- (76,170) ;
%Straight Lines [id:da7032522996496221] 
\draw    (70,90) -- (76,90) ;
%Straight Lines [id:da4029193301114602] 
\draw [line width=1.5]  [dash pattern={on 5.63pt off 4.5pt}]  (158.17,197.75) -- (313.67,197.75) ;
%Straight Lines [id:da9745170325417144] 
\draw [line width=1.5]  [dash pattern={on 5.63pt off 4.5pt}]  (158.17,102.25) -- (315.67,102.25) ;
%Curve Lines [id:da8385818372399397] 
\draw [line width=1.5]    (158.17,102.25) .. controls (180.17,87.75) and (212.8,72.5) .. (235.8,66) ;
%Curve Lines [id:da07007358359979488] 
\draw [line width=1.5]    (235.8,66) .. controls (257.8,71) and (293.17,86.25) .. (315.67,102.25) ;
%Curve Lines [id:da964065061164563] 
\draw [line width=1.5]    (313.67,197.75) .. controls (291.48,211.96) and (258.48,228.61) .. (235.4,234.8) ;
%Curve Lines [id:da48683330589101836] 
\draw [line width=1.5]    (235.4,234.8) .. controls (213.47,229.51) and (180.2,214.5) .. (157.91,198.21) ;
%Straight Lines [id:da6068007629411124] 
\draw    (155.83,243.25) -- (155.83,249.33) ;
%Straight Lines [id:da8674224606703071] 
\draw    (314.83,243.75) -- (314.83,249.83) ;

% Text Node
\draw (382,240.4) node [anchor=north west][inner sep=0.75pt]    {$\overline{q}_{0}$};
% Text Node
\draw (52,22.4) node [anchor=north west][inner sep=0.75pt]    {$q_{\tau }$};
% Text Node
\draw (46,202.4) node [anchor=north west][inner sep=0.75pt]    {$0.2$};
% Text Node
\draw (46,162.4) node [anchor=north west][inner sep=0.75pt]    {$0.4$};
% Text Node
\draw (46,122.4) node [anchor=north west][inner sep=0.75pt]    {$0.6$};
% Text Node
\draw (46,82.4) node [anchor=north west][inner sep=0.75pt]    {$0.8$};
% Text Node
\draw (46,42.4) node [anchor=north west][inner sep=0.75pt]    {$1.0$};
% Text Node
\draw (56,243.4) node [anchor=north west][inner sep=0.75pt]    {$0$};
% Text Node
\draw (308.83,253.57) node [anchor=north west][inner sep=0.75pt]    {$\overline{q}$};
% Text Node
\draw (150.83,253.4) node [anchor=north west][inner sep=0.75pt]    {$\underline{q}$};

\end{tikzpicture}
\vspace{-2em}
\caption{$\text{Supp}(q_{\tau})$ as a correspondence of $\bar{q}_{0}$}
\label{fig:1} 
\end{figure}

Figure \ref{fig:2} illustrates the engagement level $\mathbb{E}^{P}[G(q_\tau)-G(\bar{q}_0)]$ for every prior
belief. The dashed curves represent the agent-optimal policy and the
solid curves represent the principal-optimal policy. The principal-optimal
policy induces higher engagement level than the agent-optimal policy.
However, the engagement level converges to zero when the prior belief
goes to the boundary of the continuation region $(\underline{q},\bar{q})$.
\begin{figure}[htbp]
\centering 
\tikzset{every picture/.style={line width=0.75pt}} %set default line width to 0.75pt        

\begin{tikzpicture}[x=0.75pt,y=0.75pt,yscale=-1,xscale=1]

%Straight Lines [id:da32000309386167114] 
\draw    (70,250) -- (380,250) ;
%Straight Lines [id:da802301768750821] 
\draw    (70,250) -- (70,50) ;
%Straight Lines [id:da7815772409239472] 
\draw    (70,50) -- (76,50) ;
%Straight Lines [id:da04180346398878698] 
\draw    (70,130) -- (76,130) ;
%Straight Lines [id:da9282410363803599] 
\draw    (70,210) -- (76,210) ;
%Straight Lines [id:da7277286487488575] 
\draw    (70,170) -- (76,170) ;
%Straight Lines [id:da6454008756364746] 
\draw    (70,90) -- (76,90) ;
%Straight Lines [id:da27224873662538585] 
\draw    (155.83,243.58) -- (155.83,249.67) ;
%Straight Lines [id:da5840917900900985] 
\draw    (314.83,244.08) -- (314.83,250.17) ;
%Curve Lines [id:da27466149569187903] 
\draw [line width=1.5]  [dash pattern={on 5.63pt off 4.5pt}]  (155.83,249.67) .. controls (190.67,180.58) and (280.67,180.58) .. (314.83,250.17) ;
%Curve Lines [id:da767945715381495] 
\draw [line width=1.5]    (155.83,249.67) .. controls (176.17,180.75) and (195.67,138.25) .. (235.17,75.25) ;
%Curve Lines [id:da18884690860226683] 
\draw [line width=1.5]    (314.83,250.17) .. controls (294.67,180.75) and (276.67,140.25) .. (235.17,75.25) ;

% Text Node
\draw (382,240.4) node [anchor=north west][inner sep=0.75pt]    {$\overline{q}_{0}$};
% Text Node
\draw (36,22.4) node [anchor=north west][inner sep=0.75pt]    {$Engagement$};
% Text Node
\draw (46,202.4) node [anchor=north west][inner sep=0.75pt]    {$0.1$};
% Text Node
\draw (46,162.4) node [anchor=north west][inner sep=0.75pt]    {$0.2$};
% Text Node
\draw (46,122.4) node [anchor=north west][inner sep=0.75pt]    {$0.3$};
% Text Node
\draw (46,82.4) node [anchor=north west][inner sep=0.75pt]    {$0.4$};
% Text Node
\draw (46,42.4) node [anchor=north west][inner sep=0.75pt]    {$0.5$};
% Text Node
\draw (56,243.4) node [anchor=north west][inner sep=0.75pt]    {$0$};
% Text Node
\draw (308.83,253.9) node [anchor=north west][inner sep=0.75pt]    {$\overline{q}$};
% Text Node
\draw (150.83,248.73) node [anchor=north west][inner sep=0.75pt]    {$\underline{q}$};

\end{tikzpicture}
\vspace{-2em}
\caption{Engagement as a function of $\bar{q}_{0}$}
\label{fig:2} 
\end{figure}

\paragraph{Example 2: Leading Cases.}

We next adapt example 1 to illustrate our two leading cases, in the context of engaging test-motivated students. Suppose the true state
of the world $T$ consists of both a test-relevant and test-irrelevant
dimension, $T=T_{1}\times T_{2}=\{L,R\}\times\{0,1\}$. The student
and teacher know that the student will be asked a single question,
$Q=\{Q_{0}\}$, whose responses are $A=\{l,r\}$, with $Q_{0}(L0\}=Q_{0}(L1)=l$
and $Q_{0}(R0)=Q_{0}(R1)=r$. That is, the $\{L,R\}$ component of
the true state is relevant, and the $\{0,1\}$ component is irrelevant.
The student (agent) again faces a binary choice $A=\{l,r\}.$ The
student's utility is one if she answers the question correctly and
negative one otherwise, and therefore is identical (ignoring the test-irrelevant
dimension) to that in example 1.
The prior belief is uniform. We continue to assume the agent's information processing
capacity is defined by the negative of Shannon's entropy, and continue to use the parameters $ \kappa =2,\chi=1$.

\paragraph{Leading Case: $H=G$.}

We first study the leading case of $H=G$, and illustrate Proposition \ref{prop:attention} with a specific example. Like Example 1, both the principal's and the agent's optimization problems are equivalent to static rational inattention problems. By re-parameterizing the problems using the condition probability of the chosen action, both problems can be rewritten as:
\begin{align*}
\sup_{P(a|t_{1},t_{2})}\sum_{a,t_{1},t_{2}}\frac{1}{4}u(a,t_{1})P(a|t_{1},t_{2})-C\cdot\bigg(&\sum_{a}\sum_{t_{1},t_{2}}\frac{1}{4}P(a|t_{1},t_{2})\log(P(a|t_{1},t_{2})\\
&-\sum_{a}\bigg(\sum_{t_{1},t_{2}}\frac{1}{4}P(a|t_{1},t_{2})\bigg)\log\bigg(\sum_{t_{1},t_{2}}\frac{1}{4}P(a|t_{1},t_{2})\bigg)\bigg),
\end{align*}
for $C=2$ (the agent-optimal problem) and $C=2-\lambda$ (the principal-optimal problem). Note that replacing $P(a|t_{1},t_{2}=0)$
and $P(a|t_{1},t_{2}=1)$ with $\frac{P(a|t_{1},0)+P(a|t_{1},1)}{2}$
(denoted by $P(a|t_{1})$) does not change the positive term while
strictly reduces the negative term if $P(a|t_{1},0)$ and $P(a|t_{2},1)$
are not identical. Evidently, w.l.o.g., the optimization problem reduces
to

\begin{align*}
\sup_{P(a|t_{1})}\sum_{a,t_{1}}\frac{1}{2}u(a,t_{1})P(a|t_{1})-C\cdot\bigg(&\sum_{a,t_{1}}\frac{1}{2}P(a|t_{1})\log(P(a|t_{1}))\\
&-\sum_{a,t_{1}}\frac{1}{2}P(a|t_{1})\log\bigg(\sum_{t_{1}}\frac{1}{2}P(a|t_{1})\bigg)\bigg),
\end{align*}
which is equivalent to a static rational inattention problem with $t_1$ being the only unknown state.\footnote{This equivalence is special to mutual information and is the implication
of a more general ``compression invariance'' property introduced
by \citet{caplin2017shannon,bloedel2020cost}.} As shown in Proposition \ref{prop:attention}, when there is a test-irrelevant state that is costly to learn about, no information about the state will
be acquired either in the student-optimal benchmark or in the principal-agent
problem, even though the student learning about state
$t_{2}$ enters the teacher's payoff function.

\begin{figure}
\centering 
\tikzset{every picture/.style={line width=0.75pt}} %set default line width to 0.75pt        

\begin{tikzpicture}[x=0.75pt,y=0.75pt,yscale=-1,xscale=1]

%Straight Lines [id:da6692921917289287] 
\draw    (29.56,206.16) -- (29.56,96.96) ;
%Straight Lines [id:da5904693867480948] 
\draw    (209.96,205.96) -- (159.96,70.16) ;
%Straight Lines [id:da5835291404508511] 
\draw [line width=0.75]    (29.56,96.96) -- (159.96,70.16) ;
%Straight Lines [id:da3810116121552569] 
\draw    (29.56,96.96) -- (209.96,205.96) ;
%Straight Lines [id:da9371606045422223] 
\draw    (29.56,206.16) -- (209.96,205.96) ;
%Shape: Polygon [id:ds7623928573916046] 
\draw  [draw opacity=0][fill={rgb, 255:red, 177; green, 208; blue, 244 }  ,fill opacity=1 ] (94.76,83.56) -- (119.76,151.46) -- (119.76,206.06) -- (94.76,138.16) -- cycle ;
%Straight Lines [id:da5382936464098622] 
\draw    (29.56,206.16) -- (159.96,70.16) ;
%Shape: Circle [id:dp12169135584690127] 
\draw  [draw opacity=0][fill={rgb, 255:red, 208; green, 2; blue, 27 }  ,fill opacity=1 ] (107.26,144.81) .. controls (107.26,143.72) and (108.15,142.83) .. (109.24,142.83) .. controls (110.33,142.83) and (111.22,143.72) .. (111.22,144.81) .. controls (111.22,145.9) and (110.33,146.79) .. (109.24,146.79) .. controls (108.15,146.79) and (107.26,145.9) .. (107.26,144.81) -- cycle ;
%Shape: Circle [id:dp8145108812267969] 
\draw  [draw opacity=0][fill={rgb, 255:red, 0; green, 122; blue, 255 }  ,fill opacity=1 ] (71.66,148.21) .. controls (71.66,147.12) and (72.55,146.23) .. (73.64,146.23) .. controls (74.73,146.23) and (75.62,147.12) .. (75.62,148.21) .. controls (75.62,149.3) and (74.73,150.19) .. (73.64,150.19) .. controls (72.55,150.19) and (71.66,149.3) .. (71.66,148.21) -- cycle ;
%Shape: Circle [id:dp6280933135252691] 
\draw  [draw opacity=0][fill={rgb, 255:red, 0; green, 122; blue, 255 }  ,fill opacity=1 ] (140.46,142.01) .. controls (140.46,140.92) and (141.35,140.03) .. (142.44,140.03) .. controls (143.53,140.03) and (144.42,140.92) .. (144.42,142.01) .. controls (144.42,143.1) and (143.53,143.99) .. (142.44,143.99) .. controls (141.35,143.99) and (140.46,143.1) .. (140.46,142.01) -- cycle ;
%Shape: Circle [id:dp3816936232406034] 
\draw  [draw opacity=0][fill={rgb, 255:red, 126; green, 211; blue, 33 }  ,fill opacity=1 ] (40.86,150.61) .. controls (40.86,149.52) and (41.75,148.63) .. (42.84,148.63) .. controls (43.93,148.63) and (44.82,149.52) .. (44.82,150.61) .. controls (44.82,151.7) and (43.93,152.59) .. (42.84,152.59) .. controls (41.75,152.59) and (40.86,151.7) .. (40.86,150.61) -- cycle ;
%Shape: Circle [id:dp7356786978659093] 
\draw  [draw opacity=0][fill={rgb, 255:red, 126; green, 211; blue, 33 }  ,fill opacity=1 ] (166.26,139.41) .. controls (166.26,138.32) and (167.15,137.43) .. (168.24,137.43) .. controls (169.33,137.43) and (170.22,138.32) .. (170.22,139.41) .. controls (170.22,140.5) and (169.33,141.39) .. (168.24,141.39) .. controls (167.15,141.39) and (166.26,140.5) .. (166.26,139.41) -- cycle ;
%Shape: Circle [id:dp9314968397056482] 
\draw  [draw opacity=0][fill={rgb, 255:red, 208; green, 2; blue, 27 }  ,fill opacity=1 ] (238.76,100.31) .. controls (238.76,99.22) and (239.65,98.33) .. (240.74,98.33) .. controls (241.83,98.33) and (242.72,99.22) .. (242.72,100.31) .. controls (242.72,101.4) and (241.83,102.29) .. (240.74,102.29) .. controls (239.65,102.29) and (238.76,101.4) .. (238.76,100.31) -- cycle ;
%Shape: Circle [id:dp867278561911343] 
\draw  [draw opacity=0][fill={rgb, 255:red, 0; green, 122; blue, 255 }  ,fill opacity=1 ] (238.46,117.51) .. controls (238.46,116.42) and (239.35,115.53) .. (240.44,115.53) .. controls (241.53,115.53) and (242.42,116.42) .. (242.42,117.51) .. controls (242.42,118.6) and (241.53,119.49) .. (240.44,119.49) .. controls (239.35,119.49) and (238.46,118.6) .. (238.46,117.51) -- cycle ;
%Shape: Circle [id:dp7009070314262438] 
\draw  [draw opacity=0][fill={rgb, 255:red, 126; green, 211; blue, 33 }  ,fill opacity=1 ] (238.26,134.91) .. controls (238.26,133.82) and (239.15,132.93) .. (240.24,132.93) .. controls (241.33,132.93) and (242.22,133.82) .. (242.22,134.91) .. controls (242.22,136) and (241.33,136.89) .. (240.24,136.89) .. controls (239.15,136.89) and (238.26,136) .. (238.26,134.91) -- cycle ;

% Text Node
\draw (209.6,206) node [anchor=north west][inner sep=0.75pt]  [font=\small] [align=left] {R0};
% Text Node
\draw (160.4,54.4) node [anchor=north west][inner sep=0.75pt]  [font=\small] [align=left] {R1};
% Text Node
\draw (15.2,207.6) node [anchor=north west][inner sep=0.75pt]  [font=\small] [align=left] {L0};
% Text Node
\draw (10.8,89.2) node [anchor=north west][inner sep=0.75pt]  [font=\small] [align=left] {L1};
% Text Node
\draw (248,93.5) node [anchor=north west][inner sep=0.75pt]  [font=\footnotesize] [align=left] {Prior belief};
% Text Node
\draw (248,110.5) node [anchor=north west][inner sep=0.75pt]  [font=\footnotesize] [align=left] {Agent-optimal posteriors};
% Text Node
\draw (248,128.5) node [anchor=north west][inner sep=0.75pt]  [font=\footnotesize] [align=left] {Principal-optimal posteriors};

\end{tikzpicture}
\vspace{-2em}
\caption{Leading case: $H=G$.}
\label{fig:5} 
\end{figure}

Figure \ref{fig:5} illustrates the analysis above. The tetrahedron
depicts the probability simplex that lies in $\mathbb{R}^{3}$. Each
vertex of the tetrahedron is a degenerate belief (labeled by the state
that occurs with probability one). The red dot is the uniform prior.
The hyperplane represents the cutoff beliefs where the student is
indifferent between actions $l$ and $r$. To the right of the hyperplane,
$q_{R}>q_{L}$ and $r$ is optimal. Vice versa for the left of the
hyperplane. The two blue dots are the student-optimal posterior beliefs.
The two green dots are the solution to the principal-agent problem.
Evidently, all the dots are on a line segment, which illustrates
that the teacher provides only test-related information to the student. It is also evident that the teacher provides more information than the student would choose for herself.

\paragraph{Leading Case: Decision-Irrelevant $G$.}

Now, we turn to the second leading case, in which the teacher only cares about how much
the student knows about the state $t_{2}\in T_{2}$ (the test-irrelevant
dimension). We assume that the teacher's payoff is
$2\mathbb{E^{\pi}}[|q_{t_{2}}-0.5|]$ if the measure of stopping belief
is $\pi$, where $q_{t_{2}}$ denotes the probability of $t_{2}=1$
under the student's stopping belief.
%\footnote{This objective does not satisfy our differentiability assumptions for $G$ at all points, but the issues arising from non-differentiability will not be relevant for this example. } 
Note that this payoff function
is equivalent to the instrumental value of information in a hypothetical
binary decision problem where the utility of matching (mismatching)
the state $y$ is 1 (-$1$). Theorem \ref{thm:main-result} shows that the teacher's optimal policy maximizes the
following auxiliary problem (for some $\lambda>0$):

\[
\sup_{\pi}\mathbb{E^{\pi}}[2|q_{t_{2}}-0.5|+\lambda\widehat{u}(q)-2\lambda(H(q)-H(\bar{q}_{0}))],
\]

which is equivalent to a hypothetical static RI problem where the
decision maker has four possible actions, whose utilities are

\[
\begin{array}{c|cccc}
u(a,t) & L0 & R0 & L1 & R1\\
\hline a_{1} & 1+\lambda & 1-\lambda & -1+\lambda & -1-\lambda\\
a_{2} & 1-\lambda & 1+\lambda & -1-\lambda & -1+\lambda\\
a_{3} & -1+\lambda & -1-\lambda & 1+\lambda & 1-\lambda\\
a_{4} & -1-\lambda & -1+\lambda & 1-\lambda & 1+\lambda
\end{array}.
\]

Per \citet{matejka2015rational}, the solution is logit, given by

\[
\begin{array}{c|cccc}
P(a|t) & L0 & R0 & L1 & R1\\
\hline a_{1} & \frac{e^{1+\lambda}}{e^{1+\lambda}+e+e^{\lambda}+1} & \frac{e}{e^{1+\lambda}+e+e^{\lambda}+1} & \frac{e^{\lambda}}{e^{1+\lambda}+e+e^{\lambda}+1} & \frac{1}{e^{1+\lambda}+e+e^{\lambda}+1}\\
a_{2} & \frac{e}{e^{1+\lambda}+e+e^{\lambda}+1} & \frac{e^{1+\lambda}}{e^{1+\lambda}+e+e^{\lambda}+1} & \frac{1}{e^{1+\lambda}+e+e^{\lambda}+1} & \frac{e^{\lambda}}{e^{1+\lambda}+e+e^{\lambda}+1},\\
a_{3} & \frac{e^{\lambda}}{e^{1+\lambda}+e+e^{\lambda}+1} & \frac{1}{e^{1+\lambda}+e+e^{\lambda}+1} & \frac{e^{1+\lambda}}{e^{1+\lambda}+e+e^{\lambda}+1} & \frac{e}{e^{1+\lambda}+e+e^{\lambda}+1}\\
a_{4} & \frac{1}{e^{1+\lambda}+e+e^{\lambda}+1} & \frac{e}{e^{1+\lambda}+e+e^{\lambda}+1} & \frac{e}{e^{1+\lambda}+e+e^{\lambda}+1} & \frac{e^{1+\lambda}}{e^{1+\lambda}+e+e^{\lambda}+1}
\end{array}
\]

and the default rule is uniform $P(a)\equiv\frac{1}{4}.$ Note that
the conditional distribution,

\[
\begin{array}{c|cc}
P(a|t) & L0\&L1 & R0\&R1\\
\hline a_{1}\&a_{3} & \frac{e}{e+1} & \frac{1}{e+1}\\
a_{2}\&a_{4} & \frac{1}{e+1} & \frac{e}{e+1}
\end{array},
\]
is identical to the student-optimal solution (solved in Example 1). This is an illustration of the conclusion of Proposition \ref{prop:noise}.
The unknown parameter $\lambda$ can be pinned down by setting the
student's IC condition binding. The analysis above suggests that when
the teacher only cares about the test-irrelevant knowledge, she will
give the student exactly his preferred test-relevant information and
extra knowledge such that the student gets barely enough welfare to
be willing to participate.

\begin{figure}
\centering
\tikzset{every picture/.style={line width=0.75pt}} %set default line width to 0.75pt        

\begin{tikzpicture}[x=0.75pt,y=0.75pt,yscale=-1,xscale=1]

%Straight Lines [id:da8680880464186302] 
\draw    (29.56,206.16) -- (29.56,96.96) ;
%Straight Lines [id:da11640059479343412] 
\draw    (209.96,205.96) -- (159.96,70.16) ;
%Straight Lines [id:da15613638053677503] 
\draw [line width=0.75]    (29.56,96.96) -- (159.96,70.16) ;
%Straight Lines [id:da8353931828909176] 
\draw    (29.56,96.96) -- (209.96,205.96) ;
%Straight Lines [id:da523705178968483] 
\draw    (29.56,206.16) -- (209.96,205.96) ;
%Shape: Polygon [id:ds024475018286568284] 
\draw  [draw opacity=0][fill={rgb, 255:red, 177; green, 208; blue, 244 }  ,fill opacity=1 ] (94.76,83.56) -- (119.76,151.46) -- (119.76,206.06) -- (94.76,138.16) -- cycle ;
%Straight Lines [id:da8794242597428383] 
\draw    (29.56,206.16) -- (159.96,70.16) ;
%Shape: Circle [id:dp3277846223423817] 
\draw  [draw opacity=0][fill={rgb, 255:red, 208; green, 2; blue, 27 }  ,fill opacity=1 ] (107.26,144.81) .. controls (107.26,143.72) and (108.15,142.83) .. (109.24,142.83) .. controls (110.33,142.83) and (111.22,143.72) .. (111.22,144.81) .. controls (111.22,145.9) and (110.33,146.79) .. (109.24,146.79) .. controls (108.15,146.79) and (107.26,145.9) .. (107.26,144.81) -- cycle ;
%Shape: Circle [id:dp7244960124090254] 
\draw  [draw opacity=0][fill={rgb, 255:red, 0; green, 122; blue, 255 }  ,fill opacity=1 ] (71.66,148.21) .. controls (71.66,147.12) and (72.55,146.23) .. (73.64,146.23) .. controls (74.73,146.23) and (75.62,147.12) .. (75.62,148.21) .. controls (75.62,149.3) and (74.73,150.19) .. (73.64,150.19) .. controls (72.55,150.19) and (71.66,149.3) .. (71.66,148.21) -- cycle ;
%Shape: Circle [id:dp9577143340094532] 
\draw  [draw opacity=0][fill={rgb, 255:red, 0; green, 122; blue, 255 }  ,fill opacity=1 ] (140.46,142.01) .. controls (140.46,140.92) and (141.35,140.03) .. (142.44,140.03) .. controls (143.53,140.03) and (144.42,140.92) .. (144.42,142.01) .. controls (144.42,143.1) and (143.53,143.99) .. (142.44,143.99) .. controls (141.35,143.99) and (140.46,143.1) .. (140.46,142.01) -- cycle ;
%Shape: Circle [id:dp5350874386186724] 
\draw  [draw opacity=0][fill={rgb, 255:red, 126; green, 211; blue, 33 }  ,fill opacity=1 ] (66.53,120.61) .. controls (66.53,119.52) and (67.41,118.63) .. (68.51,118.63) .. controls (69.6,118.63) and (70.49,119.52) .. (70.49,120.61) .. controls (70.49,121.7) and (69.6,122.59) .. (68.51,122.59) .. controls (67.41,122.59) and (66.53,121.7) .. (66.53,120.61) -- cycle ;
%Shape: Circle [id:dp16116956834946317] 
\draw  [draw opacity=0][fill={rgb, 255:red, 126; green, 211; blue, 33 }  ,fill opacity=1 ] (149.73,174.61) .. controls (149.73,173.51) and (150.61,172.63) .. (151.71,172.63) .. controls (152.8,172.63) and (153.69,173.51) .. (153.69,174.61) .. controls (153.69,175.7) and (152.8,176.59) .. (151.71,176.59) .. controls (150.61,176.59) and (149.73,175.7) .. (149.73,174.61) -- cycle ;
%Shape: Circle [id:dp191388555466935] 
\draw  [draw opacity=0][fill={rgb, 255:red, 208; green, 2; blue, 27 }  ,fill opacity=1 ] (238.76,100.31) .. controls (238.76,99.22) and (239.65,98.33) .. (240.74,98.33) .. controls (241.83,98.33) and (242.72,99.22) .. (242.72,100.31) .. controls (242.72,101.4) and (241.83,102.29) .. (240.74,102.29) .. controls (239.65,102.29) and (238.76,101.4) .. (238.76,100.31) -- cycle ;
%Shape: Circle [id:dp49348084996826147] 
\draw  [draw opacity=0][fill={rgb, 255:red, 0; green, 122; blue, 255 }  ,fill opacity=1 ] (238.46,117.51) .. controls (238.46,116.42) and (239.35,115.53) .. (240.44,115.53) .. controls (241.53,115.53) and (242.42,116.42) .. (242.42,117.51) .. controls (242.42,118.6) and (241.53,119.49) .. (240.44,119.49) .. controls (239.35,119.49) and (238.46,118.6) .. (238.46,117.51) -- cycle ;
%Shape: Circle [id:dp41710092333428084] 
\draw  [draw opacity=0][fill={rgb, 255:red, 126; green, 211; blue, 33 }  ,fill opacity=1 ] (238.26,134.91) .. controls (238.26,133.82) and (239.15,132.93) .. (240.24,132.93) .. controls (241.33,132.93) and (242.22,133.82) .. (242.22,134.91) .. controls (242.22,136) and (241.33,136.89) .. (240.24,136.89) .. controls (239.15,136.89) and (238.26,136) .. (238.26,134.91) -- cycle ;
%Shape: Circle [id:dp1543745387503106] 
\draw  [draw opacity=0][fill={rgb, 255:red, 126; green, 211; blue, 33 }  ,fill opacity=1 ] (75.86,178.61) .. controls (75.86,177.52) and (76.75,176.63) .. (77.84,176.63) .. controls (78.93,176.63) and (79.82,177.52) .. (79.82,178.61) .. controls (79.82,179.7) and (78.93,180.59) .. (77.84,180.59) .. controls (76.75,180.59) and (75.86,179.7) .. (75.86,178.61) -- cycle ;
%Shape: Circle [id:dp9923475197622695] 
\draw  [draw opacity=0][fill={rgb, 255:red, 126; green, 211; blue, 33 }  ,fill opacity=1 ] (132.73,114.61) .. controls (132.73,113.51) and (133.61,112.63) .. (134.71,112.63) .. controls (135.8,112.63) and (136.69,113.51) .. (136.69,114.61) .. controls (136.69,115.7) and (135.8,116.59) .. (134.71,116.59) .. controls (133.61,116.59) and (132.73,115.7) .. (132.73,114.61) -- cycle ;

% Text Node
\draw (209.6,206) node [anchor=north west][inner sep=0.75pt]  [font=\small] [align=left] {R0};
% Text Node
\draw (160.4,54.4) node [anchor=north west][inner sep=0.75pt]  [font=\small] [align=left] {R1};
% Text Node
\draw (15.2,207.6) node [anchor=north west][inner sep=0.75pt]  [font=\small] [align=left] {L0};
% Text Node
\draw (10.8,89.2) node [anchor=north west][inner sep=0.75pt]  [font=\small] [align=left] {L1};
% Text Node
\draw (248,93.5) node [anchor=north west][inner sep=0.75pt]  [font=\footnotesize] [align=left] {Prior belief};
% Text Node
\draw (248,110.5) node [anchor=north west][inner sep=0.75pt]  [font=\footnotesize] [align=left] {Agent-optimal posteriors};
% Text Node
\draw (248,128.5) node [anchor=north west][inner sep=0.75pt]  [font=\footnotesize] [align=left] {Principal-optimal posteriors};

\end{tikzpicture}
\vspace{-2em}
\caption{Leading Case: Decision-Irrelevant $G$.}
\label{fig:6} 
\end{figure}

Figure \ref{fig:6} illustrates the analysis above. Except the green
dots, the figure is exactly the same as Figure \ref{fig:5}. The green
dots are the optimal posteriors of the principal-agent problem. For
any pair of green dots on the same side of the hyperplane, they equal
the blue dot in expectation. Their distance to the hyperplane is also
the same as the blue dot, illustrating that the teacher provides the
student his preferred test-relevant information together with extra
test-irrelevant knowledge.

\section{Characterization of dynamics}\label{sec:Dynamics}

In \Cref{sec:Optimal-Policy}, we showed that to solve for the optimal policy, it is sufficient to study a simple static rational inattention problem \eqref{eq:p:relaxed} and implement the dynamic policy by diluting the static solution. In this section, we seek to find necessary conditions that characterize the dynamics of the optimal policy. For the purpose of characterizing optimal polices, for this section, we focus on the non-degenerate case where the optimal policy is to acquire some information ($\bar{J}(\bar{q}_{0})>0$ in the context of Proposition \ref{prop:existence}). We show that two desiderata---interim incentive compatibility and credibility---jointly provide a tight characterization of the belief dynamics. 

\subsection{Implications of interim incentive compatibility}

To begin, we define, given a measure $\pi\in\mathcal{P}(\mathcal{P}(X))$,
the maximal value of information acquisition in a hypothetical restricted
static rational inattention problem: 
\[
\overline{V}_R(q,\pi)=\max_{\pi'\in\mathcal{P}(\text{Supp}(\pi)):E^{\pi'}[q']=q}\mathbb{E}^{\pi'}\left[\widehat{u}(q')-\widehat{u}(q)-\frac{ \kappa }{\chi}(H(q')-H(q))\right].
\]
This problem maximizes the agent's expected utility subject to the
usual constraint of Bayes consistency and the additional constraint
that the posteriors lie in the support of $\pi$. We use the notation
$\overline{V}_R(q,\pi)$ to emphasize that the problem is restricted relative
to the usual rational inattention problem. Note that the problem is
feasible only for $q$ that lie in in the convex hull of the support
of $\pi$, which we denote by $\text{Conv}(\text{Supp}(\pi))$.

We interpret $\overline{V}_R(q,\pi)$ as describing the maximum possible value of information
acquisition, in the sense that $\overline{V}_R(q,\pi)$ is the difference between
the best possible utility the agent could achieve with and without information acquisition, under the restriction that the set of stopping beliefs lie in the support of $\pi$.\footnote{The value of information acquisition is related to, but not the same
as, the notion of uncertainty defined in \citet{frankel2019quantifying}.} We define the set of ``suspensive'' beliefs given $\pi$ as the
set for which the value of information acquisition is weakly positive. 
\begin{defn}
\label{defn:comp} Given $\pi\in\mathcal{P}(\mathcal{P}(X))$,
the belief $q\in\mathcal{P}(X)$ is \emph{suspensive} if $q\text{\ensuremath{\in}Conv}(\text{Supp}(\pi))\setminus\mathrm{Supp(\pi)}$
and satisfies $\overline{V}_R(q,\pi)\geq0$; the
belief $q\in\mathcal{P}(X)$ is \emph{decisive }if $q\in\text{Supp}(\pi)$.
\end{defn}
Given a measure $\pi$ that describes the stopping beliefs, it is
intuitive that beliefs in $\mathrm{Supp(\pi)}$ are ``decisive''
because they directly lead to stopping. At the suspensive beliefs, the agent could benefit, relative to the prior consistent
with $\pi$, from the signals that the principal is offering (which
always result in some posterior in the support of $\pi$). Therefore,
suspensive beliefs are those that the principal can
possibly induce ``suspense'' before eventually leading the agent
to stopping beliefs in $\mathrm{Supp(\pi)}$.

Suppose the optimal policy in \eqref{eq:p:relaxed}, $\pi^{*}$, is
unique. In this case, the stochastic process for beliefs $q_{t}$
cannot leave the set of suspensive beliefs before
$t=\tau$. If $q_{t}$ left the convex hull of the support of $\pi^{*}$,
then $\pi^{*}$ would not describe the law of $q_{\tau}$. If $q_{t}$
was decisive or if $\overline{V}_R(q_{t},\pi^{*})<0$, the agent would choose
to stop. In the case of multiple optimal policies, this argument applies
for some optimal $\pi^{*}$, which leads to the proposition below. 
%Note that this proposition and the one that follows will include a technical caveat about measurability (that we believe is inessential), to simplify the proofs that follow.
\begin{prop}
\label{prop:comp}Given $\bar{q}_{0}\in\mathcal{P}(X)$, let $\Pi^{*}$
be the set of solutions to \eqref{eq:p:relaxed}. Suppose $(P, \tau)$ solves \eqref{eq:P} and $\sup_{\pi\in\Pi^*}\overline{V}_R(q,\pi)$ is a measurable function of $q$. Then, $\forall t\in[0,\tau)$, $q_t$ is almost surely suspensive:
\begin{align*}
\mathrm{Prob}^P\left(\sup_{\pi\in\Pi^*}\overline{V}_R(q_t,\pi)\ge 0\Big|t<\tau\right)= 1.
\end{align*}
\end{prop}
\begin{proof}
    See Appendix \ref{ssec:proof:comp}.
\end{proof}

Solving \eqref{eq:p:relaxed} generically leads to the IC constraint
binding, which suggests that the value of acquiring $\pi^{*}$ is
exactly zero at the prior for the agent. The set of suspensive beliefs thus defines beliefs that are ``more uncertain''
than the prior. Therefore, the interpretation of Proposition \ref{prop:comp}
is that two types of signals appear in any optimal policy almost surely.
The first type are \emph{decisive} signals that are so informative
that the agent makes a decision immediately after receiving the signal.
The second type are \emph{suspensive} signals that causes the agent's
beliefs to move to a posterior that is ``more uncertain'' than the
prior.

\paragraph*{Implication: The Necessity of Jumps in Beliefs.}

A direct implication of Proposition \ref{prop:comp} is that the optimal policy cannot involve
continuous sample paths for beliefs (we assume the initial prior lies
in the continuation region). If sample paths were continuous, the
belief process $q_{t}$ would have to exit the set of suspensive beliefs, which
would immediately induce the agent to stop. Consequently, beliefs
must jump discontinuously from the set of suspensive beliefs (which
lie in the continuation region of the benchmark model) to the set
of decisive beliefs, which lie in the strict stopping region of the
benchmark model (by Proposition \ref{prop:extreme-beliefs}).

To complement our discussion of the necessity of belief jumps, we analyze an extension of our model in which the belief process is required to be continuous. Formally, let $D^C\subset \mathbb{R}_+^{\mathcal{P}(X)}$ denote the set of \emph{continuous} $\mathcal{P}(X)$-valued functions, i.e., all continuous paths. Let $\mathcal{A}^{C}(\bar{q}_{0})$ denote the subset of $\mathcal{A}(\bar{q}_{0})$ whose probability measures $P$ are supported within $D^{C}$. Consider the optimization
problem: 
\begin{align}
J^{C}(\bar{q}_{0}) & =\sup_{(P,\tau)\in\mathcal{A}^{C}(\bar{q}_{0})}\mathbb{E}^{P}[G(q_{\tau})-G(\bar{q}_{0})|\mathcal{F}_{0}]\label{eqn:principle:bound}\\
\textrm{s.t.} & \tau\in\arg\max_{\tau'\in\mathcal{T}}\mathbb{E}^{P}[\widehat{u}(q_{\tau'})- \kappa \tau'|\mathcal{F}_{0}].\nonumber 
\end{align}
Equation \eqref{eqn:principle:bound} is the same as Definition \ref{def:The-principal's-problem}
except that the set of admissible strategies is restricted to have
no jumps.

We begin our analysis of this restricted problem by observing that
the agent-optimal benchmark is unchanged. This follows
from results in \citet{hebert2019rational}, who show that the agent-optimal
policy can be implemented by a pure diffusion process. Let $E^{A}=\left\{ q\in\mathcal{P}(X)|V^{B}(q)>\widehat{u}(q)\right\} $
be the continuation region of the agent-optimal benchmark and $\bar{E}^A$ be the closure of $E^A$.

As we argued previously, the beliefs $q_{t}$ must lie in the continuation
region of the agent-optimal benchmark if $t<\tau$. Given that the $\mathcal{A}^C(\bar{q}_0)$ admits only continuous processes, it follows that the stopping beliefs
must lie in $\bar{E}^A$. 
\begin{lem}
\label{lem:support} If $(P,\tau)$ is admissible in Equation
\eqref{eqn:principle:bound}, then $\mathrm{Supp}(q_{\tau})\subset \bar{E}^A$. 
\end{lem}
\begin{proof}
See the appendix, section \ref{subsec:Proof-of-Lemma-support}. 
\end{proof}
% Trivially, if the bound $d$ is sufficiently large, it does not restrict
% the principal, and the problem collapses to the problem considered
% in our main analysis. 
% \begin{prop}
% $\lim_{d\to\infty}J^{d}(\bar{q}_{0})=J(\bar{q}_{0})$. \label{prop:bound} 
% \end{prop}
% \begin{proof}
% For $d$ larger than the diameter of $\mathcal{P}(X)$, $\mathcal{A}^{d}(\bar{q}_{0})=\mathcal{A}(\bar{q}_{0})$. 
% \end{proof}
% Let us next consider the opposite case, in which the beliefs process
% must be continuous ($d=0$). In this case, the only possible stopping
% beliefs are the ones on the boundary of the continuation region in
% the agent-optimal problem ($\bar{Q}^{0}$ is the closure of $E^{A}$).
If in addition there are only two states $(|X|=2$), the locally invariant
posteriors property implies that the stopping beliefs will be identical
to those the agent would choose in the agent-optimal problem; the
only alternative the principal could choose would involve less information
acquisition by the agent.\footnote{When there are more than two states ($|X|>2$), the principal's and
agent's optimal policies need not coincide. Consider as an example
the case of two actions ($|A|=2$). The agent-optimal policy will
in this case involve a diffusion on a line segment within the probability
simplex (see \citet{hebert2019rational}). The principal cannot induce
the agent to follow this line segment beyond its endpoints, but can
send the agent signals that cause the agent's beliefs to move orthogonal
to this line segment.}
\begin{prop}
\label{prop:bound:binary} If $|X|=2$, then $J^{0}(\bar{q}_{0})$
is achieved by an agent-optimal policy. 
\end{prop}
\begin{proof}
See the appendix, section \ref{subsec:Proof-of-Proposition-bound-binary}. 
\end{proof}

This proposition illustrates that the possibility of jumps in beliefs is required to guarantee the agent's participation constraint will bind. Without jumps in the belief process, the agent's ability to stop at any time allows the agent to extract rents that would otherwise accrue to the principal.

\subsection{Implications of credibility} \label{ssec:credibility}

So far, our analysis assumes full intertemporal commitment on the part of the principal. However, in practice, commitment power might be limited, leading to a credibility concern. In this subsection, we tackle the issue of limited commitment by characterizing the necessary and sufficient conditions for a policy to be credible, i.e., immune to deviations in a sub-game perfection sense. 

The definition of sub-game perfection for a continuous time game like ours is technically delicate. To avoid these complexities, we introduce the notion of a "strongly sub-game perfect" equilibrium. An equilibrium is strongly sub-game perfect if the principal is not willing to deviate, even if he can deviate to the full-commitment solution, and likewise the agent is willing to stop even if she can instead switch to the agent-optimal benchmark at that point.

In an environment in which commitment is valuable, such an equilibrium cannot exist. We will show in our environment that some (but not all) full-commitment solutions to the principal-agent problem are strongly sub-game perfect in this sense, which is to say that commitment has no value. Nevertheless, despite this lack of value, the principal's inability to commit will prevent some full-commitment solutions from being implemented as equilibria of the game without commitment.

\begin{defn}\label{defn:seqn}
   $\forall q_0\in\mathcal{P}(X)$, $(P, \tau)\in \mathcal{A}(\bar{q}_0)$ constitutes a \emph{strongly sub-game perfect equilibrium} of the game if the agent's interim incentive compatibility constraints are satisfied and, in addition,  
    \begin{enumerate}
        \item[(i)] the principal would not deviate to the commitment solution: $\forall t$, $\mathrm{Prob}^P\Big(\mathbb{E}^P[G(q_{\tau})-G(q_t)-J(q_t)|\mathcal{F}_t]\ge 0\Big)=1$, and 
        \item[(ii)] the agent, when stopping, would not prefer to continue under the agent-optimal benchmark: $\mathrm{Prob}^P\big(V^B(q_{\tau}) \le \widehat{u}(q_{\tau})\big)=1 $.
    \end{enumerate}
    Let $\mathcal{V}:\mathcal{P}(X)\rightrightarrows \mathbb{R}_+^2$ denote the correspondence from prior to the set of strongly sub-game perfect equilibrium payoffs of the principal and agent.
\end{defn}

This notion of strongly sub-game perfect equilibrium is ``stronger'' than sub-game perfection in the sense that we overestimate deviation payoffs in conditions (i) and (ii): upon a deviation at belief, we award the principal with the full commitment payoff and the agent with her agent-optimal payoff, respectively. If such an equilibrium exists, it is clearly sub-game perfect in the usual sense as well.

In the analysis that follows, we will call a full-commitment optimal policy $(P, \tau)$ \textbf{credible} if  $(P, \tau)$ is also a strongly sub-game perfect equilibrium in the game without commitment. The following theorem shows that the optimal dilution policy we construct in \cref{sec:Optimal-Policy} is indeed credible and hence constitutes a (strongly) sub-game perfect equilibrium. To avoid technicalities, what we will not discuss is whether there are other sub-game perfect equilibria that involve lower payoffs for the principal. Let $\mathcal{P}(X)^+:=\{q\in\mathcal{P}(X)|G(q)<\sum q_{x}G(e_x)\}$, namely, the set of beliefs that information has some value for the principal.

% Evidently, $\mathcal{V}(q)\subset[0,J(q)]\times[0,V^B(q)]$. Using the notion of strong equilibrium, we define the weak equilibria.

% \begin{defn}\label{defn:weqn}
%    $\forall q_0\in\mathcal{P}(X)$, $(P, \tau)$ constitutes a \emph{weak equilibrium} of the game if
%     \begin{enumerate}
%         \item No principal deviation: $\forall t$ and $q\in\text{Supp}(q_t|t<\tau)$, $\mathbb{E}^P[G(q_{\tau})-G(q)|t<\tau]\ge \min \mathcal{V}_1(q)$;
%         \item No agent deviation: $\forall t$ and $q\in\text{Supp}(q_t|t<\tau)$, $\mathbb{E}^P[\widehat{u}(q_{\tau})- \kappa \tau|t<\tau]\ge \widehat{u}(q)- \kappa t$;
%         \item No restarting: $\forall q\in\text{Supp}(q_{\tau})$, $0\ge \max \{\min \mathcal{V}_1(q),\min \mathcal{V}_2(q)-\widehat{u}(q)\} $;
%     \end{enumerate}
% \end{defn}

% The notion of weak equilibrium is ``weaker '' than subgame perfection in the sense that we underestimate the deviation payoffs: upon any deviation, we award the players with their worst-case strong equilibrium payoff. Evidently, a strong equilibrium is also a weak equilibrium. 

\begin{thm}\label{thm:limited:commit}
    $\forall \bar{q}_0\in\mathcal{P}(X)$, $J(\bar{q}_0)$ can be achieved via a credible strategy, i.e., $(J(\bar{q}_0),0)\in\mathcal{V}(\bar{q}_0)$. Suppose $(P, \tau)$ solves \eqref{eq:P} and $\mathrm{supp}(q_{\tau}) \subset \mathcal{P}(X)^+$. $(P, \tau)$ is credible if and only if the agent's interim surplus is almost surely $0$:
    \begin{align}
        \forall t ,\quad \mathrm{Prob}^P\left( \mathbb{E}^P[\widehat{u}(q_{\tau})-\widehat{u}(q)- \kappa (\tau-t)|\mathcal{F}_t]=0\right)=1.\label{eq:no:surplus}
    \end{align}
\end{thm}
\begin{proof}
    See Appendix \ref{ssec:proof:limited:commit}.
\end{proof}

Note since $(J(\bar{q}_0),0)\in\mathcal{V}(\bar{q}_0)$, $J(\bar{q}_0)$ is itself a strongly sub-game perfect equilibrium payoff. Condition (i) in \cref{defn:seqn} must be satisfied for any weaker equilibrium notion that ``selects SPNE in all sub-games in favor of the principal.'' 
Therefore, \eqref{eq:no:surplus} is a sufficient and necessary condition that guards the principal from deviations followed by a favorable selection of continuing equilibrium.

Using this condition, we can characterize a set of beliefs that cannot be reached in any credible solution to the principal-agent problem. To do this, we will define the set of ``strictly suspensive'' beliefs. In parallel to our definition of suspensive beliefs, we start by defining, given a measure $\pi\in\mathcal{P}(\mathcal{P}(X))$,
the minimal value of information acquisition in a hypothetical restricted static rational inattention problem: 
\[
\underline{V}_R(q,\pi)=\min_{\pi'\in\mathcal{P}(\text{Supp}(\pi)):E^{\pi'}[q']=q}\mathbb{E}^{\pi'}\left[\widehat{u}(q')-\widehat{u}(q)-\frac{ \kappa }{\chi}(H(q')-H(q))\right].
\]
We will say that a belief is strictly suspensive if this value is strictly positive.

\begin{defn}
\label{defn:certain} Given $\pi\in\mathcal{P}(\mathcal{P}(X))$,
the belief $q\in\mathcal{P}(X)$ is \emph{strictly suspensive} if $q \ \text{ \ensuremath{\in} Conv}(\text{Supp}(\pi))\setminus\mathrm{Supp(\pi)}$
and satisfies $\underline{V}_R(q,\pi) > 0$.
\end{defn}

Note the contrast between suspensive beliefs and strictly suspensive beliefs. A belief is suspensive given $\pi$ if the \emph{maximum} value of information is \emph{weakly} positive, but is strictly suspensive only if the \emph{minimum} value of information is \emph{strictly} positive. 

When $\pi'=\pi$ is the only measure on the support of $\pi$ that satisfies Bayes-consistency, there is no distinction between $\overline{V}_R$ and $\underline{V}_R$.\footnote{This will be true generically when the number of actions is weakly lower than the number of states.} In this case, a belief is suspensive (resp. strictly suspensive) when this function is weakly (resp. strictly) positive. 

The following proposition demonstrates that a solution to the principal-agent problem is credible only if interim beliefs are not strictly suspensive.

\begin{prop}\label{prop:spne}
    Given $\bar{q}_{0}\in\mathcal{P}(X)$, let $\Pi^{*}$
be the set of solutions to \eqref{eq:p:relaxed}. Suppose $(P, \tau)$ solves \eqref{eq:P}, $\mathrm{supp}(q_{\tau}) \subset \mathcal{P}(X)^+$, and $\inf_{\pi\in \Pi^*}\underline{V}_R(q,\pi)$ is a measurable function of $q$. $(P, \tau)$ is credible only if $\forall t\in[0,\tau)$, $q_t$ is almost never strictly suspensive:
\begin{align*}
\mathrm{Prob}^P\left(\inf_{\pi\in \Pi^*}\underline{V}_R(q,\pi)\le0\Big| t<\tau\right)=1.
\end{align*}
\end{prop}
\begin{proof}
    See Appendix \ref{ssec:proof:spne}.
\end{proof}

Combining Propositions \ref{prop:comp} and \ref{prop:spne}, in any solution that satisfies interim incentive compatibility and credibility, interim beliefs are suspensive but not strictly suspensive. They must be suspensive, as otherwise the agent would choose to stop, but cannot be strictly suspensive, as otherwise the principal would be tempted to deviate and extract more of the surplus.

\paragraph*{Implication: Credible Beliefs Cannot Diffuse.}

When beliefs are suspensive but not strictly suspensive, the agent's value function is equal to
$\widehat{u}(q_t)$ at all times. This implies that a pure diffusion process is
infeasible almost everywhere in the continuation region. Because the
agent recognizes that the principal will leave her with no surplus,
only information that causes her to change her beliefs about the currently
optimal action is valuable. If the belief $q_{t}$ is such that one
action is strictly optimal given those beliefs (which will be true
generically), the principal must offer at least the possibility of
jumping to a different region in which another action is optimal;
otherwise, the agent will perceive no benefit from the information
provided. In the internet platform context, we interpret this result as suggesting that the principal
provides the agent with news articles containing extreme or sensational
claims, which should cause the agent to either move her beliefs a
lot or not at all, as opposed to providing more nuanced or qualified
information. In the teacher-student context, to engage a test-motivated student, the teacher must provide signals that have some possibility of changing the student's belief about the correct answer on the test.

\subsection{Illustrative Examples}

\paragraph{Example 1 (continued).}

We illustrate Propositions \ref{prop:comp} and \ref{prop:spne} by continuing  example 1.
First, focus on the $\bar{q}_{0}=\frac{1}{2}$ case. As is illustrated
in Figure \ref{fig:1}, the unique optimal $\pi^{*}$ involves two
posterior beliefs $\left\{ q^{1},q^{2}\right\} $, where $q^{1}<\underline{q}$
and $q^{2}>\bar{q}$. Per Proposition \ref{prop:implementation},
the dilution of $\pi^{*}$ implements an optimal policy. In this case,
the dilution of $\pi^{*}$ is the unique optimal policy. It is apparent
from the symmetry of the problem that the value of information acquisition,
$\overline{V}_R(q,\pi^{*})$, is maximized at $q=\frac{1}{2}$. As a result,
the set of suspensive beliefs is a singleton,
$\{\bar{q}_{0}\}$. By Proposition \ref{prop:comp}, beliefs will
remain at $\bar{q}_{0}$ until they jump to either $q^{1}$ or $q^{2}$.

When $\bar{q}_{0}<\frac{1}{2}$, the optimal policy is not unique. In Figure \ref{fig:dilution}, we plot the simplest dilution policy, where belief stays at $q_0$ until it jumps to the two optimal posteriors $q^1$ and $q^2.$ However, there can be other optimal policies. In this case, the set of suspensive beliefs is $[q_{0},1-q_{0}]$. Then, a full-commitment policy is optimal as long as beliefs remain in this interval $[q_{0},1-q_{0}]$ until eventually jumping to either
$q^{1}$ or $q^{2}$. Specifically, Figure \ref{fig:diffusion} plots one of such policies: when belief is at $q_0$ or $1-q_0$, it either jumps to the closer boundary, or drifts into the interior of $(q_0,1-q_0)$. When belief is interior, it follows a diffusion.

    \begin{figure}[htbp]
    \hspace{-0.3cm}
    \begin{minipage}{0.34\linewidth}
        \tikzset{every picture/.style={line width=0.75pt}} %set default line width to 0.75pt        

\begin{tikzpicture}[x=0.75pt,y=0.75pt,yscale=-1,xscale=1]

%Straight Lines [id:da9365440925854518] 
\draw    (10,150) -- (200,150) ;
%Straight Lines [id:da8489705870318989] 
\draw    (10,150) -- (10,145) ;
%Straight Lines [id:da42350918408525917] 
\draw    (200,150) -- (200,145) ;
%Straight Lines [id:da7116491378320837] 
\draw    (80,150) -- (80,145) ;
%Straight Lines [id:da5197120887842721] 
\draw    (130,150) -- (130,145) ;
%Straight Lines [id:da16187074455790018] 
\draw    (40,150) -- (40,145) ;
%Straight Lines [id:da4181211521973871] 
\draw    (171,150) -- (171,145) ;
%Shape: Brace [id:dp9567953215204483] 
\draw   (79.6,170.6) .. controls (79.67,175.27) and (82.04,177.56) .. (86.71,177.49) -- (94.91,177.36) .. controls (101.58,177.25) and (104.95,179.53) .. (105.02,184.2) .. controls (104.95,179.53) and (108.24,177.15) .. (114.91,177.04)(111.91,177.09) -- (123.11,176.91) .. controls (127.78,176.84) and (130.07,174.47) .. (130,169.8) ;
%Curve Lines [id:da5586499830997205] 
\draw    (80,140) .. controls (72.68,119.32) and (51.96,131.63) .. (42.37,138.31) ;
\draw [shift={(40,140)}, rotate = 324.06] [fill={rgb, 255:red, 0; green, 0; blue, 0 }  ][line width=0.08]  [draw opacity=0] (6.25,-3) -- (0,0) -- (6.25,3) -- cycle    ;
%Curve Lines [id:da5157780797491962] 
\draw    (80,140) .. controls (94.28,119.35) and (144.7,127.57) .. (167.62,138.77) ;
\draw [shift={(170,140)}, rotate = 208.65] [fill={rgb, 255:red, 0; green, 0; blue, 0 }  ][line width=0.08]  [draw opacity=0] (6.25,-3) -- (0,0) -- (6.25,3) -- cycle    ;

% Text Node
\draw (7,156.4) node [anchor=north west][inner sep=0.75pt]  [font=\footnotesize]  {$0$};
% Text Node
\draw (197,156.4) node [anchor=north west][inner sep=0.75pt]  [font=\footnotesize]  {$1$};
% Text Node
\draw (71,154.4) node [anchor=north west][inner sep=0.75pt]  [font=\footnotesize]  {$q_{0}$};
% Text Node
\draw (111,154.4) node [anchor=north west][inner sep=0.75pt]  [font=\footnotesize]  {$1-q_{0}$};
% Text Node
\draw (31,154.4) node [anchor=north west][inner sep=0.75pt]  [font=\footnotesize]  {$q_{1}$};
% Text Node
\draw (165,154.4) node [anchor=north west][inner sep=0.75pt]  [font=\footnotesize]  {$q_{2}$};
% Text Node
\draw (74,187) node [anchor=north west][inner sep=0.75pt]  [font=\tiny] [align=left] {(Strictly )Suspensive};
% Text Node
\draw (24,187) node [anchor=north west][inner sep=0.75pt]  [font=\tiny] [align=left] {Decisive};
% Text Node
\draw (158,186.8) node [anchor=north west][inner sep=0.75pt]  [font=\tiny] [align=left] {Decisive};

\end{tikzpicture}
        \vspace{-2em}
        \caption{Dilution policy}
        \label{fig:dilution}
        \end{minipage}
    \begin{minipage}{0.35\linewidth}
        \tikzset{every picture/.style={line width=0.75pt}} %set default line width to 0.75pt        

\begin{tikzpicture}[x=0.75pt,y=0.75pt,yscale=-1,xscale=1]

%Straight Lines [id:da7463472503327444] 
\draw    (10,150) -- (200,150) ;
%Straight Lines [id:da9183224640187961] 
\draw    (10,150) -- (10,145) ;
%Straight Lines [id:da4568605576410587] 
\draw    (200,150) -- (200,145) ;
%Straight Lines [id:da17279561789454045] 
\draw    (80,150) -- (80,145) ;
%Straight Lines [id:da3148095646755944] 
\draw    (130,150) -- (130,145) ;
%Straight Lines [id:da07666783735424976] 
\draw    (40,150) -- (40,145) ;
%Straight Lines [id:da536166881425336] 
\draw    (171,150) -- (171,145) ;
%Shape: Brace [id:dp42745772162663687] 
\draw   (79.6,170.6) .. controls (79.67,175.27) and (82.04,177.56) .. (86.71,177.49) -- (94.91,177.36) .. controls (101.58,177.25) and (104.95,179.53) .. (105.02,184.2) .. controls (104.95,179.53) and (108.24,177.15) .. (114.91,177.04)(111.91,177.09) -- (123.11,176.91) .. controls (127.78,176.84) and (130.07,174.47) .. (130,169.8) ;
%Curve Lines [id:da8440002288927682] 
\draw    (80,140) .. controls (72.68,119.32) and (51.96,131.63) .. (42.37,138.31) ;
\draw [shift={(40,140)}, rotate = 324.06] [fill={rgb, 255:red, 0; green, 0; blue, 0 }  ][line width=0.08]  [draw opacity=0] (6.25,-3) -- (0,0) -- (6.25,3) -- cycle    ;
%Curve Lines [id:da6116997692650341] 
\draw    (130,140) .. controls (138.41,120.32) and (154.05,127.85) .. (167.84,138.32) ;
\draw [shift={(170,140)}, rotate = 218.31] [fill={rgb, 255:red, 0; green, 0; blue, 0 }  ][line width=0.08]  [draw opacity=0] (6.25,-3) -- (0,0) -- (6.25,3) -- cycle    ;
%Shape: Resistor [id:dp9074926630519138] 
\draw   (92.86,139.79) -- (97.2,139.79) -- (98.17,136.57) -- (100.1,143) -- (102.03,136.57) -- (103.96,143) -- (105.89,136.57) -- (107.83,143) -- (109.76,136.57) -- (111.69,143) -- (112.65,139.79) -- (117,139.79) ;
%Straight Lines [id:da6652673475585289] 
\draw    (98,140) -- (91,140) ;
\draw [shift={(88,140)}, rotate = 360] [fill={rgb, 255:red, 0; green, 0; blue, 0 }  ][line width=0.08]  [draw opacity=0] (3.57,-1.72) -- (0,0) -- (3.57,1.72) -- cycle    ;
%Straight Lines [id:da9123374157744275] 
\draw    (113,140) -- (120,140) ;
\draw [shift={(123,140)}, rotate = 180] [fill={rgb, 255:red, 0; green, 0; blue, 0 }  ][line width=0.08]  [draw opacity=0] (3.57,-1.72) -- (0,0) -- (3.57,1.72) -- cycle    ;

% Text Node
\draw (7,156.4) node [anchor=north west][inner sep=0.75pt]  [font=\footnotesize]  {$0$};
% Text Node
\draw (197,156.4) node [anchor=north west][inner sep=0.75pt]  [font=\footnotesize]  {$1$};
% Text Node
\draw (71,154.4) node [anchor=north west][inner sep=0.75pt]  [font=\footnotesize]  {$q_{0}$};
% Text Node
\draw (111,154.4) node [anchor=north west][inner sep=0.75pt]  [font=\footnotesize]  {$1-q_{0}$};
% Text Node
\draw (31,154.4) node [anchor=north west][inner sep=0.75pt]  [font=\footnotesize]  {$q_{1}$};
% Text Node
\draw (165,154.4) node [anchor=north west][inner sep=0.75pt]  [font=\footnotesize]  {$q_{2}$};
% Text Node
\draw (72,187) node [anchor=north west][inner sep=0.75pt]  [font=\tiny] [align=left] {(Strictly) Suspensive};
% Text Node
\draw (24,187) node [anchor=north west][inner sep=0.75pt]  [font=\tiny] [align=left] {Decisive};
% Text Node
\draw (156,186.8) node [anchor=north west][inner sep=0.75pt]  [font=\tiny] [align=left] {Decisive};
% Text Node
\draw (85,122) node [anchor=north west][inner sep=0.75pt]  [font=\scriptsize] [align=left] {\textit{diffusion}};

\end{tikzpicture}
        \vspace{-2em}
        \caption{Jump/diffusion policy}
        \label{fig:diffusion}
        \end{minipage}
    \begin{minipage}{0.3\linewidth}
        \tikzset{every picture/.style={line width=0.75pt}} %set default line width to 0.75pt        

\begin{tikzpicture}[x=0.75pt,y=0.75pt,yscale=-1,xscale=1]

%Straight Lines [id:da9467117716547075] 
\draw    (10,150) -- (200,150) ;
%Straight Lines [id:da10973014329249753] 
\draw    (10,150) -- (10,145) ;
%Straight Lines [id:da8086660063277038] 
\draw    (200,150) -- (200,145) ;
%Straight Lines [id:da5503654926857497] 
\draw    (80,150) -- (80,145) ;
%Straight Lines [id:da42380571390350097] 
\draw    (130,150) -- (130,145) ;
%Straight Lines [id:da546367515941538] 
\draw    (40,150) -- (40,145) ;
%Straight Lines [id:da6324286788123932] 
\draw    (171,150) -- (171,145) ;
%Shape: Brace [id:dp5770460038933798] 
\draw   (79.6,170.6) .. controls (79.67,175.27) and (82.04,177.56) .. (86.71,177.49) -- (94.91,177.36) .. controls (101.58,177.25) and (104.95,179.53) .. (105.02,184.2) .. controls (104.95,179.53) and (108.24,177.15) .. (114.91,177.04)(111.91,177.09) -- (123.11,176.91) .. controls (127.78,176.84) and (130.07,174.47) .. (130,169.8) ;
%Curve Lines [id:da8470433089378424] 
\draw    (80,140) .. controls (70.6,124.56) and (50.6,133.02) .. (41.59,138.91) ;
\draw [shift={(40,140)}, rotate = 324.06] [fill={rgb, 255:red, 0; green, 0; blue, 0 }  ][line width=0.08]  [draw opacity=0] (8.4,-2.1) -- (0,0) -- (8.4,2.1) -- cycle    ;
%Curve Lines [id:da35687728812483877] 
\draw    (80,140) .. controls (89.52,113.82) and (146.24,121.8) .. (168.67,138.94) ;
\draw [shift={(170,140)}, rotate = 219.81] [fill={rgb, 255:red, 0; green, 0; blue, 0 }  ][line width=0.08]  [draw opacity=0] (8.4,-2.1) -- (0,0) -- (8.4,2.1) -- cycle    ;
%Curve Lines [id:da14205658140168942] 
\draw    (127.99,139.21) .. controls (112.54,133.53) and (93.53,135.77) .. (81.79,139.41) ;
\draw [shift={(80,140)}, rotate = 340.95] [fill={rgb, 255:red, 0; green, 0; blue, 0 }  ][line width=0.08]  [draw opacity=0] (8.4,-2.1) -- (0,0) -- (8.4,2.1) -- cycle    ;
\draw [shift={(130,140)}, rotate = 202.9] [fill={rgb, 255:red, 0; green, 0; blue, 0 }  ][line width=0.08]  [draw opacity=0] (8.4,-2.1) -- (0,0) -- (8.4,2.1) -- cycle    ;
%Curve Lines [id:da38936311121066325] 
\draw    (130,140) .. controls (136.34,125.53) and (152.53,134.86) .. (168.29,139.51) ;
\draw [shift={(170,140)}, rotate = 195.21] [fill={rgb, 255:red, 0; green, 0; blue, 0 }  ][line width=0.08]  [draw opacity=0] (8.4,-2.1) -- (0,0) -- (8.4,2.1) -- cycle    ;
%Curve Lines [id:da1525379958319406] 
\draw    (130,140) .. controls (123.04,113.4) and (56.86,121.42) .. (41.11,138.66) ;
\draw [shift={(40,140)}, rotate = 306.51] [fill={rgb, 255:red, 0; green, 0; blue, 0 }  ][line width=0.08]  [draw opacity=0] (8.4,-2.1) -- (0,0) -- (8.4,2.1) -- cycle    ;

% Text Node
\draw (7,156.4) node [anchor=north west][inner sep=0.75pt]  [font=\footnotesize]  {$0$};
% Text Node
\draw (197,156.4) node [anchor=north west][inner sep=0.75pt]  [font=\footnotesize]  {$1$};
% Text Node
\draw (71,154.4) node [anchor=north west][inner sep=0.75pt]  [font=\footnotesize]  {$q_{0}$};
% Text Node
\draw (111,154.4) node [anchor=north west][inner sep=0.75pt]  [font=\footnotesize]  {$1-q_{0}$};
% Text Node
\draw (31,154.4) node [anchor=north west][inner sep=0.75pt]  [font=\footnotesize]  {$q_{1}$};
% Text Node
\draw (165,154.4) node [anchor=north west][inner sep=0.75pt]  [font=\footnotesize]  {$q_{2}$};
% Text Node
\draw (72,187) node [anchor=north west][inner sep=0.75pt]  [font=\tiny] [align=left] {(Strictly) Suspensive};
% Text Node
\draw (24,187) node [anchor=north west][inner sep=0.75pt]  [font=\tiny] [align=left] {Decisive};
% Text Node
\draw (155,186.8) node [anchor=north west][inner sep=0.75pt]  [font=\tiny] [align=left] {Decisive};

\end{tikzpicture}
        \vspace{-2em}
        \caption{Credible policy}
        \label{fig:spne}
        \end{minipage}
    \end{figure}

However, if the principal would like the policy to be credible, the interim belief cannot be strictly suspensive. In this example, $\bar{V}_R = \underline{V}_R > 0$ for all $q_t \in (q_0, 1-q_0)$, which is to say that these beliefs are strictly suspensive. It follows that a credible policy must satisfy $q_t\in\{q_0,1-q_0\}$. Therefore, a credible belief process can only jump between $q_0$ and $1-q_0$ or jump to $q^1$ and $q^2$, as illustrated by Figure \ref{fig:spne} (note that the dilution policy is also credible, but the jump/diffusion policy is not credible). Figure \ref{fig:4} illustrates the sample paths of one credible policy, in which beliefs jump between $\bar{q}_{0}$ and $1-\bar{q}_{0}$ before eventually jumping to $q^{1}$ or $q^{2}$.
    
\begin{figure}[htbp]
\centering \includegraphics[width=0.6\textwidth]{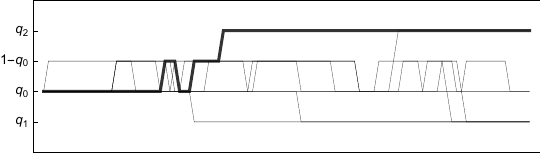} \caption{Sample paths of a credible policy}
\label{fig:4} 
\end{figure}

\paragraph{Example 2 (continued).}
We next illustrate the implications of Propositions \ref{prop:comp} and \ref{prop:spne} by continuing example 2, which discusses our two leading cases.

In the $H=G$ leading case, as discussed above, beliefs move on a line segment. As a result, even though the state space in example 2 is two-dimensional, our analysis for example 1 still applies. Because the prior is uniform, the only suspensive belief is the prior, and consequently the dilution policy is the unique optimal policy. This is illustrated below in Figure \ref{fig:cred1}.

In the decision-irrelevant $G$ leading case, in contrast, there are four possible optimal stopping beliefs under the unique $\pi^*$ that solves \eqref{eq:p:relaxed}. Because of the symmetry of the problem, we will still have $\bar{V}_R = \underline{V}_R$, and the set of beliefs that are suspensive but not strictly suspensive are characterized by the level set $\bar{V}_R(q,\pi^*)=0$. These are illustrated by the curves shown in Figure \ref{fig:cred2}.

Any credible, incentive-compatible policy must maintain interim beliefs on these curves. Because they are curves and not straight lines, this is not compatible with beliefs following a diffusion process. But is compatible with a wide range of jump processes, which are the multi-dimensional analogs of the credible policies illustrated in Figures \ref{fig:spne} and \ref{fig:4}.

    \begin{figure}[htbp]
    \begin{minipage}{0.49\linewidth}
    \centering
        \tikzset{every picture/.style={line width=0.75pt}} %set default line width to 0.75pt        

\begin{tikzpicture}[x=0.75pt,y=0.75pt,yscale=-1,xscale=1]

%Straight Lines [id:da12206137747393941] 
\draw [color={rgb, 255:red, 189; green, 16; blue, 224 }  ,draw opacity=1 ]   (98.53,231.83) -- (112.27,231.83) ;
%Straight Lines [id:da980033164689261] 
\draw    (84.76,210.16) -- (84.76,100.96) ;
%Straight Lines [id:da9455844554295831] 
\draw    (265.16,209.96) -- (215.16,74.16) ;
%Straight Lines [id:da9850668855359934] 
\draw [line width=0.75]    (84.76,100.96) -- (215.16,74.16) ;
%Straight Lines [id:da18881311938920975] 
\draw    (84.76,100.96) -- (265.16,209.96) ;
%Straight Lines [id:da8842419600815814] 
\draw    (84.76,210.16) -- (265.16,209.96) ;
%Shape: Polygon [id:ds9445976185869494] 
\draw  [draw opacity=0][fill={rgb, 255:red, 177; green, 208; blue, 244 }  ,fill opacity=1 ] (149.96,87.56) -- (174.96,155.46) -- (174.96,210.06) -- (149.96,142.16) -- cycle ;
%Straight Lines [id:da9370214922457911] 
\draw    (84.76,210.16) -- (215.16,74.16) ;
%Shape: Circle [id:dp2663414897970011] 
\draw  [draw opacity=0][fill={rgb, 255:red, 189; green, 16; blue, 224 }  ,fill opacity=1 ] (162.46,148.81) .. controls (162.46,147.72) and (163.35,146.83) .. (164.44,146.83) .. controls (165.53,146.83) and (166.42,147.72) .. (166.42,148.81) .. controls (166.42,149.9) and (165.53,150.79) .. (164.44,150.79) .. controls (163.35,150.79) and (162.46,149.9) .. (162.46,148.81) -- cycle ;
%Shape: Circle [id:dp9718922763202904] 
\draw  [draw opacity=0][fill={rgb, 255:red, 0; green, 122; blue, 255 }  ,fill opacity=1 ] (126.86,152.21) .. controls (126.86,151.12) and (127.75,150.23) .. (128.84,150.23) .. controls (129.93,150.23) and (130.82,151.12) .. (130.82,152.21) .. controls (130.82,153.3) and (129.93,154.19) .. (128.84,154.19) .. controls (127.75,154.19) and (126.86,153.3) .. (126.86,152.21) -- cycle ;
%Shape: Circle [id:dp8402708079533492] 
\draw  [draw opacity=0][fill={rgb, 255:red, 0; green, 122; blue, 255 }  ,fill opacity=1 ] (195.66,146.01) .. controls (195.66,144.92) and (196.55,144.03) .. (197.64,144.03) .. controls (198.73,144.03) and (199.62,144.92) .. (199.62,146.01) .. controls (199.62,147.1) and (198.73,147.99) .. (197.64,147.99) .. controls (196.55,147.99) and (195.66,147.1) .. (195.66,146.01) -- cycle ;
%Shape: Circle [id:dp594620075045901] 
\draw  [draw opacity=0][fill={rgb, 255:red, 126; green, 211; blue, 33 }  ,fill opacity=1 ] (96.06,154.61) .. controls (96.06,153.52) and (96.95,152.63) .. (98.04,152.63) .. controls (99.13,152.63) and (100.02,153.52) .. (100.02,154.61) .. controls (100.02,155.7) and (99.13,156.59) .. (98.04,156.59) .. controls (96.95,156.59) and (96.06,155.7) .. (96.06,154.61) -- cycle ;
%Shape: Circle [id:dp8237918157953685] 
\draw  [draw opacity=0][fill={rgb, 255:red, 126; green, 211; blue, 33 }  ,fill opacity=1 ] (221.46,143.41) .. controls (221.46,142.32) and (222.35,141.43) .. (223.44,141.43) .. controls (224.53,141.43) and (225.42,142.32) .. (225.42,143.41) .. controls (225.42,144.5) and (224.53,145.39) .. (223.44,145.39) .. controls (222.35,145.39) and (221.46,144.5) .. (221.46,143.41) -- cycle ;

% Text Node
\draw (115.8,225.67) node [anchor=north west][inner sep=0.75pt]  [font=\footnotesize] [align=left] {Credible interim beliefs};
% Text Node
\draw (264.8,210) node [anchor=north west][inner sep=0.75pt]  [font=\small] [align=left] {R0};
% Text Node
\draw (215.6,58.4) node [anchor=north west][inner sep=0.75pt]  [font=\small] [align=left] {R1};
% Text Node
\draw (70.4,211.6) node [anchor=north west][inner sep=0.75pt]  [font=\small] [align=left] {L0};
% Text Node
\draw (66,93.2) node [anchor=north west][inner sep=0.75pt]  [font=\small] [align=left] {L1};

\end{tikzpicture}
        \vspace{-2em}
        \caption{Credible interim beliefs\\ when $H=G$}
        \label{fig:cred1}
        \end{minipage}
    \begin{minipage}{0.49\linewidth}
    \centering
        \tikzset{every picture/.style={line width=0.75pt}} %set default line width to 0.75pt        

\begin{tikzpicture}[x=0.75pt,y=0.75pt,yscale=-1,xscale=1]

%Curve Lines [id:da4649849561644894] 
\draw [color={rgb, 255:red, 189; green, 16; blue, 224 }  ,draw opacity=1 ]   (131.34,158.11) .. controls (148.23,147.13) and (165.25,132.82) .. (162.74,124.31) .. controls (160.23,115.8) and (138.9,103.47) .. (122.01,100.11) ;
%Straight Lines [id:da15670681185963753] 
\draw    (83.06,185.66) -- (83.06,76.46) ;
%Straight Lines [id:da9810038377610208] 
\draw    (263.46,185.46) -- (213.46,49.66) ;
%Straight Lines [id:da7313260324420285] 
\draw [line width=0.75]    (83.06,76.46) -- (213.46,49.66) ;
%Straight Lines [id:da9138735131864657] 
\draw    (83.06,76.46) -- (263.46,185.46) ;
%Straight Lines [id:da0587460239228248] 
\draw    (83.06,185.66) -- (263.46,185.46) ;
%Shape: Polygon [id:ds5712802953884795] 
\draw  [draw opacity=0][fill={rgb, 255:red, 177; green, 208; blue, 244 }  ,fill opacity=0.81 ] (148.26,63.06) -- (173.26,130.96) -- (173.26,185.56) -- (148.26,117.66) -- cycle ;
%Straight Lines [id:da5803846165705335] 
\draw    (83.06,185.66) -- (213.46,49.66) ;
%Shape: Circle [id:dp6591558074962252] 
\draw  [draw opacity=0][fill={rgb, 255:red, 208; green, 2; blue, 27 }  ,fill opacity=1 ] (160.76,124.31) .. controls (160.76,123.22) and (161.65,122.33) .. (162.74,122.33) .. controls (163.83,122.33) and (164.72,123.22) .. (164.72,124.31) .. controls (164.72,125.4) and (163.83,126.29) .. (162.74,126.29) .. controls (161.65,126.29) and (160.76,125.4) .. (160.76,124.31) -- cycle ;
%Shape: Circle [id:dp42103854934118945] 
\draw  [draw opacity=0][fill={rgb, 255:red, 0; green, 122; blue, 255 }  ,fill opacity=1 ] (125.16,127.71) .. controls (125.16,126.62) and (126.05,125.73) .. (127.14,125.73) .. controls (128.23,125.73) and (129.12,126.62) .. (129.12,127.71) .. controls (129.12,128.8) and (128.23,129.69) .. (127.14,129.69) .. controls (126.05,129.69) and (125.16,128.8) .. (125.16,127.71) -- cycle ;
%Shape: Circle [id:dp8357437953899198] 
\draw  [draw opacity=0][fill={rgb, 255:red, 0; green, 122; blue, 255 }  ,fill opacity=1 ] (193.96,121.51) .. controls (193.96,120.42) and (194.85,119.53) .. (195.94,119.53) .. controls (197.03,119.53) and (197.92,120.42) .. (197.92,121.51) .. controls (197.92,122.6) and (197.03,123.49) .. (195.94,123.49) .. controls (194.85,123.49) and (193.96,122.6) .. (193.96,121.51) -- cycle ;
%Shape: Circle [id:dp04706736757115082] 
\draw  [draw opacity=0][fill={rgb, 255:red, 126; green, 211; blue, 33 }  ,fill opacity=1 ] (120.03,100.11) .. controls (120.03,99.02) and (120.91,98.13) .. (122.01,98.13) .. controls (123.1,98.13) and (123.99,99.02) .. (123.99,100.11) .. controls (123.99,101.2) and (123.1,102.09) .. (122.01,102.09) .. controls (120.91,102.09) and (120.03,101.2) .. (120.03,100.11) -- cycle ;
%Shape: Circle [id:dp35367280405981694] 
\draw  [draw opacity=0][fill={rgb, 255:red, 126; green, 211; blue, 33 }  ,fill opacity=1 ] (203.23,154.11) .. controls (203.23,153.01) and (204.11,152.13) .. (205.21,152.13) .. controls (206.3,152.13) and (207.19,153.01) .. (207.19,154.11) .. controls (207.19,155.2) and (206.3,156.09) .. (205.21,156.09) .. controls (204.11,156.09) and (203.23,155.2) .. (203.23,154.11) -- cycle ;
%Shape: Circle [id:dp7845598004606986] 
\draw  [draw opacity=0][fill={rgb, 255:red, 126; green, 211; blue, 33 }  ,fill opacity=1 ] (129.36,158.11) .. controls (129.36,157.02) and (130.25,156.13) .. (131.34,156.13) .. controls (132.43,156.13) and (133.32,157.02) .. (133.32,158.11) .. controls (133.32,159.2) and (132.43,160.09) .. (131.34,160.09) .. controls (130.25,160.09) and (129.36,159.2) .. (129.36,158.11) -- cycle ;
%Shape: Circle [id:dp9361014952714146] 
\draw  [draw opacity=0][fill={rgb, 255:red, 126; green, 211; blue, 33 }  ,fill opacity=1 ] (186.23,94.11) .. controls (186.23,93.01) and (187.11,92.13) .. (188.21,92.13) .. controls (189.3,92.13) and (190.19,93.01) .. (190.19,94.11) .. controls (190.19,95.2) and (189.3,96.09) .. (188.21,96.09) .. controls (187.11,96.09) and (186.23,95.2) .. (186.23,94.11) -- cycle ;
%Curve Lines [id:da23471213663784196] 
\draw [color={rgb, 255:red, 189; green, 16; blue, 224 }  ,draw opacity=1 ]   (205.21,154.11) .. controls (189.9,148.8) and (165.25,132.82) .. (162.74,124.31) .. controls (160.23,115.8) and (178.29,98.84) .. (188.21,94.11) ;
%Straight Lines [id:da4337522670253575] 
\draw [color={rgb, 255:red, 189; green, 16; blue, 224 }  ,draw opacity=1 ]   (101.03,206.83) -- (114.77,206.83) ;

% Text Node
\draw (263.1,185.5) node [anchor=north west][inner sep=0.75pt]  [font=\small] [align=left] {R0};
% Text Node
\draw (213.9,33.9) node [anchor=north west][inner sep=0.75pt]  [font=\small] [align=left] {R1};
% Text Node
\draw (68.7,187.1) node [anchor=north west][inner sep=0.75pt]  [font=\small] [align=left] {L0};
% Text Node
\draw (64.3,68.7) node [anchor=north west][inner sep=0.75pt]  [font=\small] [align=left] {L1};
% Text Node
\draw (118.3,200.67) node [anchor=north west][inner sep=0.75pt]  [font=\footnotesize] [align=left] {Credible interim beliefs};

\end{tikzpicture}
        \vspace{-2em}
        \caption{Credible interim beliefs with a decision-irrelevant $G$ }
        \label{fig:cred2}
        \end{minipage}
    \end{figure}

\section{\label{subsec:Optimal-Policy-without}Extension: Optimal Policy without Capacity
Constraints}

In this section, we consider a modified version of our model with
a constant cost of delay in which the agent has an unlimited capacity
to acquire information, and the principal's goal is to maximize the
time spent by the agent on the platform. This modified model is the
continuous-time analog of the discrete-time models studied by \citet{jan2020dynamic}
and \citet{koh2022attention}. The purpose of this section is to illustrate
the connection between these models and a special case of our more
general framework. The agent-optimal policy in the modified model
is for the agent to learn the optimal action with certainty immediately,
as this avoids entirely the cost of delay.

The principal in this case chooses his policies from the set $\mathcal{\bar{A}}(\bar{q}_{0})$,
which is the set of probability measures on $(\Omega,\mathcal{F})$
such that $q$ is martingale belief processes with $q_{0}=\bar{q}_{0}$
and non-negative stopping times $\tau$. Note that this set does not
impose the constraint on the rate of information acquisition, \eqref{eq:general-constraint},
that was imposed in our main analysis. The principal solves 
\begin{align*}
J(\bar{q}_{0}) & =\sup_{(P,\tau)\in\bar{\mathcal{A}}(\bar{q}_{0})}\mathbb{E}^{P}[\tau|\mathcal{F}_{0}]
\end{align*}
subject to the same constraint with respect to the agent's stopping
decision, 
\[
\tau\in\arg\max_{\tau'\in\mathcal{T}}\mathbb{E}^{P}[\widehat{u}(q_{\tau'})-\int_{0}^{\tau'} \kappa dt|\mathcal{F}_{0}].
\]

The first part of Lemma \ref{lem:neccesary} remains applicable: if
$\pi$ is the law of $q_{\tau}$, we must have $\mathbb{E}^{\pi}[\widehat{u}(q)-\widehat{u}(\bar{q}_{0})]\geq \kappa \mathbb{E}^{P}[\tau|\mathcal{F}_{0}]$.
Now observe that the expected utility under $\pi$ is bounded above
by the utility of fully learning the state. Let $\pi^{\max}\in\mathcal{P}(X)$
be the unique probability measure that places full support on the
extreme points of $\mathcal{P}(X)$ (i.e. the $e_{x}$ basis vectors),
with $\pi^{max}(e_{x})=\bar{q}_{0,x}$.

By the convexity of $\widehat{u}$, for all $\pi$ such that $E^{\pi}[q]=\bar{q}_{0}$,
\[
\mathbb{E}^{\pi^{max}}[\widehat{u}(q)-\widehat{u}(\bar{q}_{0})]\geq\mathbb{E}^{\pi}[\widehat{u}(q)-\widehat{u}(\bar{q}_{0})]\geq \kappa \mathbb{E}^{P}[\tau|\mathcal{F}_{0}]\geq \kappa J(\bar{q}_{0}).
\]
It follows that $ \kappa ^{-1}\mathbb{E}^{\pi^{max}}[\widehat{u}(q)-\widehat{u}(\bar{q}_{0})]$
is an upper bound on the utility achievable in this problem.

But now observe that the $\alpha$-dilution of $\pi^{\max}$, as defined
in \eqref{eq:Poisson}, with intensity $\alpha= \kappa (\mathbb{E}^{\pi^{max}}[\widehat{u}(q)-\widehat{u}(\bar{q}_{0})])^{-1}$,
achieves this bound. Moreover, the policy is incentive-compatible:
at each instant, the agent compares the utility benefit of the signal's
arrival, $\alpha\mathbb{E}^{\pi^{max}}[\widehat{u}(q)-\widehat{u}(\bar{q}_{0})]$,
against the cost of delay, $ \kappa $, and is willing to continue.
It follows that this policy is optimal.

This policy is the optimal policy in a special case of our main model.
Consider the case of our main model in which $ \kappa =\chi$
and $G(q)=H(q)=\widehat{u}(q)$. In this special case, $\pi^{\max}$ is
an optimal policy in our relaxed problem \eqref{eq:p:relaxed}, because
the constraint in that problem satisfied for any policy, and the process
described in Proposition \ref{prop:implementation} is exactly the
process above. The intuition behind this equivalence is that in our
main model, it is always optimal for the principal to exhaust the
agent's information processing capacity. If exhausting this capacity
necessarily involves satisfying the incentive compatibility constraint,
then the principal is free to choose the process that simultaneously
exhausts the agent's capacity and takes as long as possible to reach
a decisive belief.

\section{\label{sec:Conclusion}Conclusion}

We have considered the problem of a principal who provides information
to an agent so as to maximize the attention the agent pays to the
principal's information (engagement). The agent values this information
for instrumental purposes, is rational and Bayesian, and faces a constraint
on the rate at which she can process information. Our main results
are (i) that by maximizing engagement, the principal leaves the agent
no better off than if she could not receive any information at all,
and (ii) that the agent will end up holding extreme beliefs, relative
to a benchmark in which the agent could choose the information for
herself, and (iii) that optimal signals in the absence of commitment cannot either increase or decrease the value of future information prior to stopping (i.e. that they are suspensive but not strictly suspensive). Our results highlight the distortions created when an information provided maximizes engagement.

\fontsize{11pt}{12pt}\selectfont

\setlength{\bibsep}{0pt}
\bibliographystyle{plainnat}
\bibliography{references_comp}

\linespread{1.2}
\normalsize

\appendix
%dummy comment inserted by tex2lyx to ensure that this paragraph is not empty
\setlength\abovedisplayskip{2pt}
\setlength\belowdisplayskip{2pt}
\section{Omitted Proofs for Section \ref{sec:Optimal-Policy}}

\subsection{\label{ssec:proof:necessary}Proof of Lemma \ref{lem:neccesary}}
\begin{proof}
Condition (1): let $\tau'\equiv0$. Under this policy, the agent's utility from
stopping is $\widehat{u}(\bar{q}_0)$. The optimality of $\tau$ and $\mathbb{E}^\pi [\hat{u}(q)] = \mathbb{E}^P [\hat{u}(q_\tau)]$ implies
condition (1).

Condition (2): This follows from $\mathbb{E}^\pi [H(q)] = \mathbb{E}^P [H(q_\tau)]$, the optional stopping theorem, and the supermartingale property of $H(q_t)-\chi t$. The last two imply that $\mathbb{E}^P [H(q_\tau)-\chi \tau ] \leq H(q_0)$.
\end{proof}
\subsection{\label{ssec:proof:existence}Proof of Proposition \ref{prop:existence}}
\begin{proof}
      Assumption \ref{ass:imperfect} implies that the constraint in \eqref{eq:p:relaxed} must be binding, as otherwise, the optimal $\pi^*$ will achieve the principal's first best, leaving the agent with negative surplus. Then, \eqref{eq:p:relaxed} is a convex optimization problem that satisfies the conditions in Theorem 4 of \citet{zhong2018information}, which immediately implies the existence of $\lambda$ and of a finite support solution $\pi^*$. 

    Next, we prove $\mathbb{E}^{\pi^*}[\bar{J}(q)]=0$. Suppose for the purpose of contradiction that $\mathbb{E}^{\pi^*}[J(q)]>0$. For each $q\in\text{Supp}(\pi^*)$, let $\pi^q$ be a probability measure that implements $\bar{J}(q)$. Define $\pi'\in \Pi(\bar{q}_0)$ by  $d\pi'(q')=\mathbb{E}^{\pi^*(q)}[d\pi^q(q')]$, observing that 
    \[
    \mathbb{E}^{\pi'(q')}[q'] = \mathbb{E}^{\pi^*(q)}[\mathbb{E}^{\pi^q(q')}[q']]=\mathbb{E}^{\pi^*(q)}[q] = \bar{q}_0.
    \]
    Then,
    \begin{align*}
        \mathbb{E}^{\pi'(q')}[G(q')-G(\bar{q}_0)]=&\mathbb{E}^{\pi^*(q)}[\mathbb{E}^{\pi^{q}(q')}[G(q')]-G(q)+G(q)-G(\bar{q}_0)]\\
        =&\mathbb{E}^{\pi^*(q)}[\mathbb{E}^{\pi^{q}(q')}[G(q')-G(q)]]+\bar{J}(\bar{q}_0)\\
        =&\mathbb{E}^{\pi^*(q)}[\bar{J}(q)]+\bar{J}(\bar{q}_0)\\
        >&\bar{J}(\bar{q}_0).
    \end{align*}
    Meanwhile,
    \begin{align*}
        \mathbb{E}^{\pi'(q')}\left[\widehat{u}(q')-\frac{ \kappa }{\chi}H(q')-\widehat{u}(\bar{q}_0)+\frac{ \kappa }{\chi}H(\bar{q}_0)\right]=&\mathbb{E}^{\pi^*(q)}\left[\mathbb{E}^{\pi^{q}(q')}\left[\widehat{u}(q')-\frac{ \kappa }{\chi}H(q')\right]-\widehat{u}(\bar{q}_0)+\frac{ \kappa }{\chi}H(\bar{q}_0)\right]\\
        =&\mathbb{E}^{\pi^*(q)}\left[\mathbb{E}^{\pi^{q}(q')}\left[\widehat{u}(q')-\frac{ \kappa }{\chi}H(q')-\widehat{u}(q)+\frac{ \kappa }{\chi}H(q)\right]\right]\\
        &+\mathbb{E}^{\pi^*(q)}\left[\widehat{u}(q)-\frac{ \kappa }{\chi}H(q)-\widehat{u}(\bar{q}_0)+\frac{ \kappa }{\chi}H(\bar{q}_0)\right]\\
        \ge & 0.
    \end{align*}
    Therefore, $\pi'$ is feasible in \eqref{eq:p:relaxed} and strictly improves on $\pi^*$, contradicting the optimality of $\pi^*$. It follows by the non-negativity of $J(q)$ that $\mathbb{E}^{\pi^*}[J(q)]=0$.
\end{proof}

\subsection{\label{ssec:proof:implementation}Proof of Proposition \ref{prop:implementation}}
\begin{proof}
Under the $\alpha$-dilution policy, if $t < \tau$,
\begin{align*}
\mathbb{E}^P[H(q_{t+\delta})-\chi(t+\delta)|\mathcal F_t,\tau>t] &= (1-e^{-\alpha \delta}) \mathbb{E}^\pi[H(q) - H(\bar{q}_0)]  + H(\bar{q}_0) -\chi(t+\delta).
\end{align*}
Using $\alpha = \chi (\mathbb{E}^\pi[H(q) - H(\bar{q}_0)])^{-1}$ and $x \geq 1-e^{-x}$,
\[
\mathbb{E}^P[H(q_{t+\delta})-\chi(t+\delta)|\mathcal F_t,\tau>t] \leq H(\bar{q}_0) -\chi(t).
\] as required by \eqref{eq:general-constraint}. That constraint holds by assumption if $t \geq \tau$.

By $q_{\tau}\sim\pi$, the principal's utility is $\mathbb{E}^{\pi}[G(q)-G(\bar{q}_{0})]$. What remains to be verified is the agent's optimality condition.

Take any admissible stopping time $\tau'$, 
\begin{align*}
\mathbb{E}^{P}[\widehat{u}(q_{\tau'})- \kappa \tau'|\mathcal{F}_{0}]= & \mathrm{Prob}(\tau'<\tau |\mathcal{F}_{0})\left(\widehat{u}(\bar{q}_0)- \kappa \mathbb{E}^{P}\left[\tau'|\tau'<\tau, \mathcal{F}_{0}\right]\right)\\
 & +\mathrm{Prob}(\tau'\ge\tau |\mathcal{F}_{0})\left(\mathbb{E}^{P}\left[\widehat{u}(q_{\tau})- \kappa \tau'|\tau'\ge\tau, \mathcal{F}_{0}\right]\right)\\
\le & \mathrm{Prob}(\tau'<\tau|\mathcal{F}_{0})\left(\mathbb{E}^{\pi}[\widehat{u}(q)]- \kappa \mathbb{E}^{P}[\tau|\mathcal{F}_{0}]- \kappa \mathbb{E}^{P}\left[\tau'|\tau'<\tau, \mathcal{F}_{0}\right]\right)\\
 & +\mathrm{Prob}(\tau'\ge\tau|\mathcal{F}_{0})\left(\mathbb{E}^{\pi}[\widehat{u}(q)]- \kappa \mathbb{E}^{P}\left[\tau|\tau'\ge\tau, \mathcal{F}_{0}\right]\right)\\
= & \mathrm{Prob}(\tau'<\tau | \mathcal{F}_{0})\left(\mathbb{E}^{\pi}[\widehat{u}(q)- \kappa \mathbb{E}^{P}\left[\tau|\tau'<\tau, \mathcal{F}_{0}\right]\right)\\
 & +\mathrm{Prob}(\tau'\ge\tau | \mathcal{F}_{0})\left(\mathbb{E}^{\pi}[\widehat{u}(q)]- \kappa \mathbb{E}^{P}\left[\tau|\tau'\ge\tau, \mathcal{F}_{0}\right]\right)\\
= & \mathbb{E}^{P}[\widehat{u}(q)- \kappa \tau|\mathcal{F}_{0}].
\end{align*}
The first equality is from the definition of conditional expectations, the process defining $q_{t}$, equation \eqref{eq:Poisson}, and the convention that $q_t =q_\tau$ for $t>\tau$. The
first inequality is from the constraint $\mathbb{E}^{\pi}[\widehat{u}(q)-\widehat{u}(\bar{q}_{0})]\ge \kappa \mathbb{E}^{P}[\tau|\mathcal{F}_{0}]$
in the relaxed problem \eqref{eq:p:relaxed} and $E[\tau|\tau'\geq\tau]\leq E[\tau'|\tau'\geq\tau]$.
The second equality is from the memorylessness property of $\tau$,\footnote{Specifically, $\mathbb{E}^{P}\left[\tau|\tau'<\tau, \mathcal{F}_{0}\right]=\mathbb{E}^{P}\left[\tau-\tau'|\tau'<\tau, \mathcal{F}_{0}\right]+\mathbb{E}^{P}\left[\tau'|\tau'<\tau, \mathcal{F}_{0}\right]=\mathbb{E}^{P}\left[\tau|\mathcal{F}_{0}\right]+\mathbb{E}^{P}\left[\tau'|\tau'<\tau,\mathcal{F}_{0}\right]$.}
and the last from the definition of conditional expectations. Note
that if the constraint in \eqref{eq:p:relaxed} is slack, the first
inequality is strict if $\tau'$ is not equal to $\tau$ $P$-a.e.
In this case, the $\alpha$-dilution of $\pi$ strongly implements
principal's utility level $\mathbb{E}^{\pi}[G(q)-G(\bar{q}_{0})]$. 
\end{proof}

\subsection{\label{ssec:proof:agent}Proof of \Cref{prop:agent}}

\begin{proof}
    If $V^B(\bar{q}_0)>\widehat{u}(\bar{q}_0)$, then let $\pi$ implement $V^B(\bar{q}_0)$. By the assumption that full information is strictly beneficial, let $\pi'$ be the distribution of posterior beliefs corresponding to full information and note by Assumption $\mathbb{E}^{\pi'}[G(q')-G(\bar{q}_0)]>0$. Then, the dilution of a convex combination of $\pi$ and $\pi'$ is feasible, incentive compatible for sufficiently small weight on $\pi'$ and strictly profitable for the principal. Therefore, $J(\bar{q}_0)>0$. 

    Next, suppose for the purpose of contradiction that  $\mathbb{E}^{\pi_A}[V^B(q)-\widehat{u}(q)]>0$.  For each $q\in\text{Supp}(\pi_A)$, let $\pi^q$ be the probability measure that implements $V^B(q)$. Define $\pi' \in \Pi(\bar{q}_0)$ by $d\pi'(q')=\mathbb{E}^{\pi_A(q)}[d\pi^q(q')]$. Then     \begin{align*}
        \mathbb{E}^{\pi'(q')}\left[\widehat{u}(q')-\frac{ \kappa }{\chi}H(q')-\widehat{u}(\bar{q}_0)+\frac{ \kappa }{\chi}H(\bar{q}_0)\right]=&\mathbb{E}^{\pi_A(q)}\left[\mathbb{E}^{\pi^q(q')}\left[\widehat{u}(q')-\frac{ \kappa }{\chi}H(q')\right]-\widehat{u}(\bar{q}_0)+\frac{ \kappa }{\chi}H(\bar{q}_0)\right]\\
        =&\mathbb{E}^{\pi_A(q)}\left[\mathbb{E}^{\pi^q(q')}\left[\widehat{u}(q')-\frac{ \kappa }{\chi}H(q')-\widehat{u}(q)+\frac{ \kappa }{\chi}H(q)\right]\right]\\
        &+\mathbb{E}^{\pi_A(q)}\left[\widehat{u}(q)-\frac{ \kappa }{\chi}H(q)-\widehat{u}(\bar{q}_0)+\frac{ \kappa }{\chi}H(\bar{q}_0)\right]\\
        =& \mathbb{E}^{\pi_A(q)} \left[ V^B(q) - \hat{u}(q) \right] + V^B(\bar{q}_0) -\widehat{u}(\bar{q}_0) \\
        >& V^B(\bar{q}_0)-\widehat{u}(\bar{q}_0).
    \end{align*}
    This contradicts the optimality of $\pi_A$, and by $V^B(q)\geq \hat{u}(q)$ we must have $\mathbb{E}^{\pi_A}[V^B(q)-\hat{u}(q)]=0$.
\end{proof}

\subsection{\label{subsec:proof:extreme-belief}Proof of Proposition \ref{prop:extreme-beliefs}}
\begin{proof}
Observe first by Lemma \ref{prop:agent} that 
$ V^B (\bar{q}_{0})>\widehat{u}(\bar{q}_0)$
implies $\lambda>0$ (where $\lambda$ is defined as in Proposition \ref{prop:existence}).

Define 
\begin{align*}
U^{a}(q)=\widehat{u}(q)-\frac{ \kappa }{\chi}H(q)
\end{align*}
and define 
\begin{align*}
U^{p}(q)=\widehat{u}(q)-\frac{ \kappa }{\chi}H(q)+\lambda G(q).
\end{align*}
Define $H^{*}(q)$ as the concave envelope of $\frac{ \kappa }{\chi}H(q)-\lambda G(q)$.
Observe that any
\[
\pi^{*}\in\arg\max_{\pi\in\mathcal{P}(\mathcal{P}(X)):\mathbb{E}^{\pi}[q]=\bar{q}_{0}}\mathbb{E}^{\pi}\left[\widehat{u}(q)-\frac{ \kappa }{\chi}H(q)+\lambda G(q)\right]
\]
places full support on points where $H^{*}(q)=\frac{ \kappa }{\chi}H(q)-\lambda G(q)$,
by the convexity of $\widehat{u}(q)$, and for each $q'\in\mathcal{P}(X)$
there exists a $\pi'\in\mathcal{P}(\mathcal{P}(X))$ with $\mathbb{E}^{\pi'}[q]=q'$
and 
\[
\widehat{u}(q')-H^{*}(q')\leq\mathbb{E}^{\pi'}[\widehat{u}(q)-\frac{ \kappa }{\chi}H(q)+\lambda G(q)],
\]
and therefore 
\[
\pi^{*}\in\arg\max_{\pi\in\mathcal{P}(\mathcal{P}(X)):\mathbb{E}^{\pi}[q]=\bar{q}_{0}}\mathbb{E}^{\pi}\left[\widehat{u}(q)-H^{*}(q)\right].
\]

Now observe that $q^{*}\in Q^{p}(\bar{q}_{0})$ and $ J (\bar{q}_{0})>0$
implies by Proposition \ref{prop:existence} and Theorem \ref{thm:main-result}
that there exists a $\pi^{*}\in\mathcal{P}(\mathcal{P}(X))$ with
$q^{*}\in\mathrm{Supp}\pi^{*}$ such that $\pi^{*}$ is a solution
to a static rational inattention problem. By the Lagrangian lemma
of \citet{caplin2017shannon} applied to the problem with the cost
function $H^{*}$, there exists a linear function $L^{p}:\mathcal{P}(X)\rightarrow\mathbb{R}$
such that 
\[
\widehat{u}(q)-H^{*}(q)-L^{p}(q)\leq\widehat{u}(q^{*})-H^{*}(q^{*})-L^{p}(q^{*})
\]
with equality for all $q\in\mathrm{Supp}\pi^{*}$. By $H^{*}(q)\geq\frac{ \kappa }{\chi}H(q)-\lambda G(q)$,
with equality for all $q\in\mathrm{Supp}\pi^{*}$, 
\[
U^{p}(q)-L^{p}(q)\leq U^{p}(q^{*})-L^{p}(q^{*}),
\]
with equality for all $q\in\mathrm{Supp}\pi^{*}$.

Note that $\widehat{u}$ is piecewise linear and convex and both $H$
and $G$ are continuously differentiable. It follows that $\widehat{u}$
can not involve a convex kink at any $q\in\mathrm{Supp}\pi^{*}$,
and hence if such a $q$ lies in the relative interior of $\mathcal{P}(X)$,
\[
\nabla U^{p}(q)=\nabla L^{p}(q).
\]

The same conclusions apply for any $q^{+}\in Q^{a}(\bar{q}_{0})$:
there must exist an associated $\pi^{+}\in\mathcal{P}(\mathcal{P}(X))$
and linear function $L^{a}(q)$ such that for all $q$, 
\[
U^{a}(q)-L^{a}(q)\leq U^{a}(q^{*})-L^{a}(q^{*}),
\]
with equality for all $q\in\mathrm{Supp}\pi^{+}$, and with $\nabla U^{p}(q)=\nabla L^{p}(q)$
for all such $q$ in the relative interior of $\mathcal{P}(X)$.

We first prove that, under the assumption that $Q^{a}(\bar{q}_{0})\subseteq\mathrm{relint}(\mathcal{P}(X))$,
there cannot exist a $q^{*}\in Q^{p}(\bar{q}_{0})\cap Q^{a}(\bar{q}_{0})$.
Proof by contradiction: suppose such a $q^{*}$ exists, as part of
a solution to the relaxed principal's problem $\pi^{*}\in\mathcal{P}(\mathcal{P}(X))$.
We must have, by the linearity of the functions $L^{a}$ and $L^{p}$,
\[
\nabla L^{p}(q)-\nabla L^{a}(q)=\nabla G(q^{*})
\]
for all $q\in\mathcal{P}(X)$. Therefore, 
\begin{align*}
\mathbb{E}^{\pi^{*}}[U^{p}(q)] & =U^{p}(q^{*})+(\bar{q}_{0}-q^{*})(\nabla G(q^{*})+L^{a}(q^{*}))\\
 & =U^{a}(q^{*})+G(q^{*})+(\bar{q}_{0}-q^{*})(\nabla G(q^{*})+L^{q}(q^{*}))\\
 & =\mathbb{E}^{\pi^{*}}[U^{a}(q)+q\cdot\nabla G(q^{*})]+G(q^{*})-q^{*}\cdot\nabla G(q^{*})\\
 & \leq\mathbb{E}^{\pi^{*}}[U^{a}(q)+G(q)]=\mathbb{E}^{\pi^{*}}[U^{p}(q)].
\end{align*}
Thus, we must have 
\[
G(q^{*})+\mathbb{E}^{\pi^{*}}[(q-q^{*})\nabla G(q^{*})]=\mathbb{E}^{\pi^{*}}[G(q)],
\]
but this requires (by $ J (\bar{q}_{0})=E^{\pi^{*}}[G(q)]-G(\bar{q}_{0})]>0$)
that
\[
G(q^{*})+(\bar{q}_{0}-q^{*})\nabla G(q^{*})>G(\bar{q}_{0}),
\]
contradicting the convexity of $G$. We conclude that $q^{*}\in Q^{p}(\bar{q}_{0})$
implies $q^{*}\notin Q^{a}(\bar{q}_{0})$.

Now suppose $q^{*}\in Q^{p}(\bar{q}_{0})\cap\mathrm{Conv}Q^{a}(\bar{q}_{0})$.
Definitionally, there is some $\pi'\in\mathcal{P}(Q^{a}(\bar{q}_{0}))$
such that $\mathbb{E}^{\pi'}[q]=q^{*}$, and by $q^{*}\notin Q^{a}(\bar{q}_{0})$,
\[
U^{a}(q^{*})<\mathbb{E}^{\pi'}[U^{a}(q)].
\]
It follows by the convexity of $G$ that 
\[
U^{p}(q^{*})<\mathbb{E}^{\pi'}[U^{a}(q)+G(q^{*})]\leq\mathbb{E}^{\pi'}[U^{p}(q)],
\]
which contradicts $q^{*}\in Q^{p}(\bar{q}_{0})$. We conclude that
$q^{*}\in Q^{p}(\bar{q}_{0})$ implies $q^{*}\notin\mathrm{Conv}Q^{a}(\bar{q}_{0})$
\end{proof}

\subsection{\label{ssec:proof:attention}Proof of \Cref{prop:attention}}
\begin{proof}
    Per Proposition \ref{prop:existence}, the optimal $\pi$ necessarily leads to $J(\bar{q}_0)=\mathbb{E}^{\pi}[\widehat{u}(q)-\widehat{u}(\bar{q}_0)]$. Then, \eqref{eq:p:relaxed} reduces to: 
    \begin{align*}
\sup_{\substack{\pi\in\mathcal{P}(\mathcal{P}(X)):\mathbb{E}^{\pi}[q]=q_{0}}
} & \mathbb{E}^{\pi}[\widehat{u}(q)-\widehat{u}(\bar{q}_{0})]\\
s.t.\  & \frac{ \kappa }{\chi}\mathbb{E}^{\pi}[H(q)-H(\bar{q}_{0})]\le\mathbb{E}^{\pi}[\widehat{u}(q)-\widehat{u}(\bar{q}_{0})].
\end{align*}

Proposition \ref{prop:existence} implies that $\pi$ could be picked with finite support. Let $s$ be the signal that induces $\pi$ per the Bayes rule. Let $p(s|x_1,x_2)$ denote the conditional distribution of $s$. Let $s'$ be another signal that has conditional distribution $p(s'=z|x_1,x_2)=\mathbb{E}^{\bar{q}_0}[p(s=z|x_1,x_2)|x_1]$, i.e., $s'$ contains $s$'s information about $x_1$ only. Let $\pi'$ be the distribution of posterior induced by $s'$. Evidently, since $u$ depends on $x_1$ only, 
\begin{align*}
    \mathbb{E}^{\pi}[\widehat{u}(q)-\widehat{u}(\bar{q}_{0})]=\mathbb{E}^{\pi_1}[\widehat{u}(q,\bar{q}_0|_{X_2})-\widehat{u}(\bar{q}_{0})]=\mathbb{E}^{\pi'}[\widehat{u}(q)-\widehat{u}(\bar{q}_{0})].
\end{align*}
Therefore, $\pi'$ is feasible. Next, we show that $\pi'$ weakly improves $\pi$. Note that 
\begin{align*}
    &\mathbb{E}^{\pi}[H(q)-H(\bar{q}_0)]=I(s;x_1,x_2)=\mathbb{E}^{\bar{q}_0}[D_{KL}(p(s|x_1,x_2))||p(s)];\\
    &\mathbb{E}^{\pi'}[H(q)-H(\bar{q}_0)]=I(s';x_1,x_2)=\mathbb{E}^{\bar{q}_0}[D_{KL}(p(s'|x_1))||p(s')]\\
    \implies & \mathbb{E}^{\pi}[H(q)-H(\bar{q}_0)] - \mathbb{E}^{\pi'}[H(q)-H(\bar{q}_0)]\\
    &=\mathbb{E}^{\bar{q}_0}\left[\log\frac{p(s|x_1,x_2)}{p(s)}-\log\frac{p(s|x_1)}{p(s)}\right]\\
    &=\mathbb{E}^{\bar{q}_0}[D_{KL}(p(s|x_1,x_2)||p(s|x_1))]\ge 0,
\end{align*}
strictly if $p(s|x_1,x_2)\neq p(s|x_1)$ for some $(x_1,x_2)$ in the support of $\bar{q}_0$. 

\end{proof}

\subsection{\label{ssec:proof:noise}Proof of \Cref{prop:noise}}

\begin{proof}
    We slightly abuse notation and let $x_1,x_2$ denote the random variables that determines the state realization. $\forall \pi$, let $\pi_i=\pi|_{X_i}$. Let $s_1$ be the random variable that is independent to $x_2$ and induces posterior belief $\pi_1$ for $x_1$. Let $s_2$ be the random variable that induces posterior belief $\pi$. Let $I$ denote mutual information. Then,
    \begin{align*}
        \mathbb{E}^{\pi}[H(q)-H(\bar{q}_{0})]=&I(s_1,s_2;x_1,x_2)\\
        =&I(s_1,s_2;x_1)+I(s_1,s_2;x_2|x_1)\\
        \ge&I(s_1;x_2)+I(s_2;x_2|x_1)\\
        =& I(s_1;x_2)+I(s_2,x_1;x_2)-\underbrace{I(x_1;x_2)}_{=0}\\
        \ge&I(s_1;x_2)+I(s_2;x_2)\\
        =&\mathbb{E}^{\pi_1}[H(q,\bar{q}_0|_{X_2})-H(\bar{q}_{0})]+\mathbb{E}^{\pi_2}[H(\bar{q}_0|_{X_1},q)-H(\bar{q}_{0})].
    \end{align*}
    In the derivation, $(q,\bar{q}_0|_{x_2})$ denotes the belief vector where the probability of state $x_1$ is given by $q$ and the probability of state $x_2$ is given by $\bar{q}_0$ restricted to $x_2$, respectively. Belief vector $(\bar{q}_0|_{X_1},q)$ is defined analogously. 
    Therefore, let $\pi'=\pi_1\otimes\pi_2$, 
    \begin{align*}
        &\mathbb{E}^{\pi}[H(q)-H(\bar{q}_{0})]\ge \mathbb{E}^{\pi_1}[H(q,\bar{q}_0|_{X_2})-H(\bar{q}_{0})]+\mathbb{E}^{\pi_2}[H(\bar{q}_0|_{X_1},q)-H(\bar{q}_{0})]= \mathbb{E}^{\pi'}[H(q)-H(\bar{q}_{0})];\\
        &\mathbb{E}^{\pi}[\widehat{u}(q)-\widehat{u}(\bar{q}_{0})]=\mathbb{E}^{\pi_1}[\widehat{u}(q,\bar{q}_0|_{X_2})-\widehat{u}(\bar{q}_{0})]=\mathbb{E}^{\pi'}[\widehat{u}(q)-\widehat{u}(\bar{q}_{0})];\\
        &\mathbb{E}^{\pi}[G(q)-G(\bar{q}_{0})]=\mathbb{E}^{\pi_2}[G(\bar{q}_0|_{X_1},q)-G(\bar{q}_{0})]=\mathbb{E}^{\pi'}[G(q)-G(\bar{q}_{0})].
    \end{align*}
    Then, $\pi'$ is feasible and weakly improves $\pi$. Proposition \ref{prop:existence} implies that $J(\bar{q}_0)$ can be solved by maximizing
    \begin{align*}
        &\mathbb{E}^{\pi_1\otimes\pi_2}\left[\widehat{u}(q)-\frac{ \kappa }{\chi}H(q)+\lambda G(q)\right]\\
        =&\underbrace{\mathbb{E}^{\pi_1}\left[\widehat{u}(q,\bar{q}_0|_{X_2})-\frac{ \kappa }{\chi}H(q,\bar{q}_0|_{X_2})\right]}_{\text{maximized by $\pi^*_1$ solving $V^B$}}+\mathbb{E}^{\pi_2}\left[-\frac{ \kappa }{\chi}H(\bar{q}_0|_{X_1},q)+\lambda G(\bar{q}_0|_{X_1},q)\right].
    \end{align*}
\end{proof}

\section{Omitted Proofs for \Cref{sec:Dynamics}}
\subsection{\label{ssec:proof:comp}Proof of \cref{prop:comp}}
\begin{proof}
By Theorem \ref{thm:main-result}, there is some $\pi\in\Pi^{*}$
that is the law of $q_{\tau}$ under the optimal policy. By the martingale property of the belief
process, $q_{t}=E^P[q_{\tau}|\mathcal{F}_{t},t<\tau]$ for
all $t\in[0,\tau)$, and thus $q_{t}$ must lie in the convex
hull of the support of $\pi$; hence, $\overline{V}_R(q_t,\pi)$ is well defined.

For the purpose of contradiction, suppose that the inequality in Proposition \ref{prop:comp} does not hold. Then, there exists $\epsilon>0$ and a positive probability event $\Omega'\in \mathcal{F}_t|_{t<\tau}$ s.t. $\forall \omega\in \Omega'$, $\overline{V}_R(q_t(\omega),\pi)<-\epsilon$. By the definition of $\overline{V}$,
\begin{align*}
    \mathbb{E}\left[\widehat{u}(q_{\tau})-\widehat{u}(q_t)-\frac{ \kappa }{\chi}(H(q_{\tau})-H(q_t))\Big| t<\tau,\omega\in\Omega'\right]\le -\epsilon.
\end{align*}
However, this violates the incentive compatibility condition of the agent. By deviating to stopping immediately at event $\Omega'$, the agent's expected utility improves by $\epsilon\cdot\mathrm{Prob}^P(\Omega')$. Contradiction.
\end{proof}

\subsection{\label{subsec:Proof-of-Lemma-support}Proof of Lemma \ref{lem:support}}
\begin{proof}
Define: 
\begin{align*}
\tau'=\tau\wedge q_{t}\text{ first leaves }E^{A}.
\end{align*}
By definition $\tau'\le\tau$. Let $\bar{E}^A$ denote the closure of $E^A$. Since $\mathrm{Supp(q_{\tau'-})}\subset E^{A}$, $\mathrm{Supp}(q_{\tau'})\subset \bar{E}^A$.
We prove by contradiction that $\tau'=\tau$. Suppose $\tau'<\tau$
on a positive measure, on which 
\begin{align*}
\mathbb{E}^{P}\left[\widehat{u}(q_{\tau})-(\tau-\tau')\cdot \kappa \big|\tau'<\tau\right]= & \mathbb{E}^{P}\left[\mathbb{E}^{P}\left[\widehat{u}(q_{\tau})-(\tau-\tau')\cdot \kappa \big|\tau'\right]\big|\tau'<\tau\right]\\
< & \mathbb{E}^{P}\left[\widehat{u}(q_{\tau'})\big|\tau'<\tau\right].
\end{align*}
The inequality is from the fact that $\tau'<\tau\implies q_{\tau'}\not\in E^{A}$.
Therefore, $\tau'$ strictly improves upon $\tau$; hence, $\tau$ is
not incentive compatible. The contradiction implies that $\mathrm{Supp}(q_{\tau})\subset\bar{E}^A$. 
\end{proof}

\subsection{\label{subsec:Proof-of-Proposition-bound-binary}Proof of Proposition
\ref{prop:bound:binary}}
\begin{proof}
As is discussed in Section \ref{sec:Optimal-Policy}, the agent-optimal
policy can be solved by concavifying $\widehat{u}(q)-\frac{ \kappa }{\chi}H(q)$.
Therefore, there exists a linear function $L(q)$ that is weakly higher
than $\widehat{u}(q)-\frac{ \kappa }{\chi}H(q)$ and tangents
it at two beliefs $q^{1}<\bar{q}_{0}<q^{2}$. WLOG, let $q^{1}$ and
$q^{2}$ be the smaller and largest such beliefs, respectively. Since
$\widehat{u}-\frac{ \kappa }{\chi}H$ is piece-wise strictly concave,
the interval $[q^{1},q^{2}]$ is bounded away from the rest of $E^{A}$.

Proposition \ref{lem:support} implies that any admissible principal's
strategy has $\mathrm{Supp}(q_{\tau})\subset\bar{E}^{A}$. Moreover, any continuous path that starts
from $\bar{q}_{0}$ and ends outside of $[q^{1},q^{2}]$ leaves $E^{A}$;
hence, it is not admissible. Therefore, $\mathrm{Supp}(q_{\tau})\subset[q^{1},q^{2}]$.

Therefore, $J^{0}(\bar{q}^{0})$ is bounded above by the following
relaxed problem: 
\begin{align*}
\sup_{\pi\in\mathcal{P}(\mathcal{P}([q^{1},q^{2}]))} & \mathbb{E}^{\pi}[\rho(G(q)-G(\bar{q}_{0}))].
\end{align*}

Suppose $G$ is affine on $[q^1,q^2]$, then the principal is completely indifferent between any strategy; hence, the agent optimal strategy is optimal for the principal as well. If $G$ is not affine on $[q^1,q^2]$, the relaxed problem is solved by $\pi^{*}$
with support $\left\{ q^{1},q^{2}\right\} $ (such $\pi^{*}$ is unique).
By \citet{hebert2019rational}, there exists a Gaussian process that
implements $\pi^{*}$ and satisfies the information constraint. Note
that this Gaussian process also implements the agent maximal continuation
payoff (the upper concave hull of $\widehat{u}(q)-\frac{ \kappa }{\chi}(H(q)-H(q_{t}))$)
for every interim belief; hence, it is incentive compatible. 
\end{proof}

\subsection{\label{ssec:proof:limited:commit}Proof of \cref{thm:limited:commit}}
\begin{proof}
    \textbf{Sufficiency}: condition (ii) of Definition \ref{defn:seqn} is a direct implication of Propositions \ref{prop:existence} and \ref{prop:agent}. Next, we verify condition (i). Suppose for the purpose of contradiction that condition (i) is violated, i.e. there exist some $t$, a positive probability Borel set $B\in\mathcal{F}_t$ and $\epsilon>0$ s.t. $\forall B'\subset B$ and $B'\in \mathcal{F}_t$,
    \begin{align*}
        \mathbb{E}^P[G(q_{\tau})-G(q_t)|\omega\in B']\le \mathbb{E}^P[J(q_t)|\omega\in B']-\epsilon;
    \end{align*}

    Let $\pi'_q$ be the policy that implements $J(q)$ for each $q$. Let $\pi$ be the distribution of $q_{\tau}$. Then, the inequality above implies 
    \begin{align*}
        \mathbb{E}^{P}[\mathbb{E}^{\pi'_{q_t(\omega)}}[G(q)|\omega\in B]P(B)+\mathbb{E}^P[G(q_{\tau})|\omega\not\in B](1-P(B))> \mathbb{E}^{\pi}[G(q)].
    \end{align*}
    Meanwhile, \eqref{eq:no:surplus} implies that
    \begin{align*}
        \mathbb{E}^P[\widehat{u}(q_{\tau})-\widehat{u}(q_t)-\kappa(\tau-t)|\omega\in B']=0.
    \end{align*}    
    Since each $\pi_q$ is incentive compatible itself, 
    \begin{align*}
        &\mathbb{E}^{P}[\mathbb{E}^{\pi'_{q_t(\omega)}}[\widehat{u}(q)-\widehat{u}(q_t(\omega))-\frac{\kappa}{\chi}(H(q)-H(q_t(\omega))]|\omega\in B]\\
        \ge& 0=\mathbb{E}^P[\widehat{u}(q_{\tau})-\widehat{u}(q_t)-\frac{\kappa}{\chi}(H(q_{\tau})-H(q_t))|\omega\in B]\\
        \implies \ \ &\mathbb{E}^{P}[\mathbb{E}^{\pi'_{q_t(\omega)}}[\widehat{u}(q)-\frac{\kappa}{\chi}H(q)]|\omega\in B]\ge \mathbb{E}^P[\widehat{u}(q_{\tau})-\frac{\kappa}{\chi}H(q_{\tau})|\omega\in B]
    \end{align*}

The two inequalities above implies that the probability measure $\pi''=\mathbb{E}^P[\pi'_{q_t(\omega)}|\omega\in B]P(B)+\mathbb{E}^P[\delta_{q_{\tau}}|\omega\not\in B](1-P(B))$ is feasible in \eqref{eq:p:relaxed} and strictly improves upon $\pi$, leading to a contradiction.

Due to the sufficiency result, the dilution policy defined in Theorem \ref{thm:main-result} specifies a strong equilibrium.

\textbf{Necessity}: Suppose for the purpose of contradiction that the equality doesn't hold. Then, there exist some $t$, a positive probability Borel set $B\in\mathcal{F}_t$ and $\epsilon>0$ s.t. $\forall B'\subset B$ and $B'\in \mathcal{F}_t$,
\begin{align*}
    &\mathbb{E}^P[\widehat{u}(q_{\tau})-\widehat{u}(q_t)- \kappa (\tau-t)|\omega\in B']\ge \epsilon\\
    \implies&\mathbb{E}^P\left[\widehat{u}(q_{\tau})-\widehat{u}(q_t)-\frac{\kappa}{\chi}(H(q_{\tau})-H(q_t))|\omega\in B'\right]\ge \epsilon.
\end{align*}

Let $\bar{I}=\max_{q\in\mathcal{P}(X)}(\sum q_x H(e_x)-H(q))$. Let $\delta=\bar{I}/(\bar{I}+\epsilon)$. Let $\pi'_q$ be the distribution of $q_{\tau}|_{q_t=t}$ for $q\in\mathrm{supp}q_t$. Let $\pi''_q=\delta\pi'_q+(1-\delta) \sum q_x e_x$. Then, by construction, $\forall B'\subset B$,
\begin{align*}
\mathbb{E}^P\left[\mathbb{E}^{\pi''_{q_t(\omega)}}\left[\widehat{u}(q)-\widehat{u}(q_t)-\frac{\kappa}{\chi}(H(q)-H(q_t))\right]\Big|\omega\in B'\right]\ge 0.
\end{align*}
Since $\pi''_{q}$ leaves the agent with non-negative utility, it is feasible in \eqref{eq:p:relaxed}; hence, $\mathbb{E}^{\pi''_{q_t(\omega)}}[G(q)-G(q_t(\omega))]\le J(q_t(\omega))$.
By assumption, $\mathrm{supp}(q_{\tau})\subset \mathcal{P}(X)^+$; hence, $\mathrm{supp}(q_t)\subset \mathcal{P}(X)^+$ due to the martingale property. Therefore, $\forall B'\subset B$, $\mathbb{E}^P[G(q_{\tau})|\omega\in B']<\mathbb{E}^P[\sum q_x(\omega) G(e_x)|\omega\in B']$. This implies,
\begin{align*}
\mathbb{E}^P[G(q_{\tau})-G(q_t(\omega))|\omega\in B']<\mathbb{E}^P\left[\mathbb{E}^{\pi''_{q_t(\omega)}}\left[G(q)-G(q_t(\omega))\right]\big|\omega\in B'\right]\le \mathbb{E}^P[J(q_t)|\omega\in B'].
\end{align*}
This violates condition (i), which states 
\begin{align*}
    \mathrm{Prob}^P\Big(\mathbb{E}^P[G(q_{\tau})-G(q_t)-J(q_t)|\mathcal{F}_t]\ge 0\Big)=1.
\end{align*}
\end{proof}

\subsection{Proof of \cref{prop:spne}\label{ssec:proof:spne}}
\begin{proof}
By Theorem \ref{thm:main-result}, there is some $\pi\in\Pi^*$
that is the law of $q_{\tau}$ under the optimal policies. By the martingale property of the belief
process, $q_{t}=E^{P}[q_{\tau}|\mathcal{F}_{t},t<\tau]$ for
all $t\in[0,\tau)$, and thus $q_{t}$ must lie in the convex
hull of the support of $\pi$. 

For the sake of contradiction, suppose Proposition \ref{prop:spne} does not hold. Then, there exists positive probability event $B'\in\mathcal{F}_t|_{t<\tau}$ and $\epsilon>0$ such that $\forall \omega\in B'$,
\begin{align*}
    \inf_{\pi\in \Pi^*}\underline{V}_R(q_t(\omega),\pi)\ge\epsilon.
\end{align*}
Then, $\mathbb{E}^P\left[\widehat{u}(q_{\tau})-\widehat{u}(q_t)-\frac{ \kappa }{\chi}(H(q_{\tau})-H(q_t))\Big|\omega\in B'\right]\ge \epsilon,$
i.e., the agent obtains positive interim surplus on positive probability event $B'$, which violates \cref{thm:limited:commit}.
\end{proof}

\end{document}